\documentclass[a4paper, 11pt]{article}

\usepackage{fullpage}

\usepackage{lineno,hyperref}
\modulolinenumbers[5]

\usepackage{amsmath}
\usepackage{amsthm}
\usepackage{amsfonts}
\usepackage{array}
\usepackage{multirow}
\usepackage{float}

\newtheorem{thm}{Theorem}
\newtheorem{lem}[thm]{Lemma}
\newtheorem{fact}[thm]{Fact}
\newtheorem{claim}[thm]{Claim}
\newtheorem{corollary}[thm]{Corollary}

\newcommand\floor[1]{\lfloor#1\rfloor}
\newcommand\ceil[1]{\lceil#1\rceil}

\usepackage{tcolorbox}
\usepackage{color,soul}
\usepackage{xcolor}

\usepackage{authblk}

\newcount\Comments  % 0 suppresses notes to selves in text
\definecolor{darkgreen}{rgb}{0,0.6,0}
\newcommand{\kibitz}[2]{\ifnum\Comments=1{\color{#1}{#2}}\fi}

\usepackage{todonotes}

\title{An extension of the Moran process using \\ type-specific connection graphs}

\author[1]{Themistoklis Melissourgos}
\author[2,3]{Sotiris E. Nikoletseas}
\author[2,3]{Christoforos L. Raptopoulos}
\author[3,4]{Paul G. Spirakis}
\affil[1]{Group of Operations Research, TU Munich, Germany}
\affil[2]{Computer Technology Institute \& Press ``Diophantus'' (CTI), Patras, Greece}
\affil[3]{Computer Engineering \& Informatics Department, University of Patras, Greece}
\affil[4]{Department of Computer Science, University of Liverpool, United Kingdom}

\date{}

\begin{document}

\maketitle

		\begin{abstract}
			In the Moran process, as studied by Lieberman, Hauert and Nowak \cite{L05}, agents of a two-type population (i.e. mutants and residents) are associated with the nodes of a graph.  Initially, there is only one mutant in the population, occupying a node chosen uniformly at random, while all other individuals are residents. The reproductive power of each type is measured by its relative fitness: mutants have fitness $r > 0$, while residents have fitness $1$. In every step, an individual is chosen with probability proportional to its fitness, and its type (mutant or resident) is passed on to a neighbor which is chosen uniformly at random. It has been suggested that this stochastic process can model the spread of genetic mutations in populations. This work concerns a generalization of the model of Lieberman et al. \cite{L05}, which we introduced in \cite{MNRS18}. In particular, we assume that different types of individuals perceive the population through different graphs defined on the same vertex set, namely the resident graph and the mutant graph. We study structural and algorithmic aspects related to the fixation probability, which is the probability that eventually only mutants remain in the population.
			
			In the first part of this work, we show how we can transfer known results from the original single-graph model of \cite{L05} to our two-graphs model. In that direction, we provide a generalization of the Isothermal Theorem of \cite{L05}, thus providing sufficient conditions for a pair of graphs to have fixation probability equal to that of a pair of cliques; this corresponds to the absorption probability of a birth-death process with forward bias $r$. 
			
			In the second part, we give a 2-player strategic game view where each player corresponds to a different type and we associate fixation and/or extinction probabilities with player payoffs. In particular, pure strategies in this game are graphs and we give evidence that the clique is the most beneficial for both players, by proving bounds on the fixation probability when one of the two graphs is complete and the other graph belongs to various natural graph classes.

			%In this setting, we attempt to identify best responses for each player. We give evidence that the clique is the most beneficial graph for both players, by proving bounds on the fixation probability when one of the two graphs is complete and the other graph belongs to various natural graph classes.  
			
			Finally, we study the problem of efficiently approximating the fixation probability. We initially show that there is a pair of graphs for which the fixation probability is exponentially small. Therefore, the fixation probability in the general case of an arbitrary pair of graphs cannot be approximated via a method similar to \cite{D14}. Nevertheless, in the special case when the mutant graph is complete, we describe an FPRAS for the approximation of the fixation probability.
			
			%an efficient approximation of the fixation probability is possible through an FPRAS which we describe.
		\end{abstract}

	\section{Introduction}
	
	The Moran process \cite{M58} models the antagonism between two species whose critical difference in terms of adaptation is their \textit{relative fitness}. A \textit{resident} has relative fitness 1 and a \textit{mutant} relative fitness $r>0$. Many settings in Evolutionary Game Theory consider fitness as a measure of reproductive success; for examples see \cite{N06,HJSK,EDKJ}. A generalization of the Moran process by Lieberman et al \cite{L05} considered the situation where the replication of an individual's fitness depends on some given structure, i.e. a directed graph. This model gave rise to an extensive line of works in Computer Science with motivation from computer networks and social networks.
	%
	%In this work we further extend the model of \cite{L05} to capture the situation where, different species perceive the population through different graphs that determine their way of spreading their offspring. As in the single graph case, the restrictions of the process imply that eventually only one species will remain in the population. Our setting is by definition an interaction between two players (species) that want to maximize their probability of occupying the whole population.
	
	In this work we further extend the model of \cite{L05} to capture the situation where different species (types) perceive the population through different graphs defined on the same vertex set. In particular, the type (mutant or resident) of a node determines the outgoing edges incident to it. In other words, the type of the node determines the candidates (neighbours), one of which will inherit the node's type in case it reproduces. As in the single-graph case, the restrictions of the process imply that eventually only one species will remain in the population. Our setting is an interaction between two players (species) where each wants to maximize her probability of occupying the whole population.
	
	This strategic interaction is described by an 1-sum bimatrix game, where each player (resident or mutant) has all the strongly connected digraphs on $n$ nodes as her pure strategies. The resident's payoff is the \textit{extinction probability} (i.e. the probability that eventually only residents remain in the population) and the mutant's payoff is the \textit{fixation probability} (which is complementary to the extinction probability). A basic question that interests us is: what are the pure Nash equilibria of this game (if any)? To gain a better understanding of the behaviour of the competing graphs, we investigate the best responses of the resident to the clique graph of the mutant.
	
	This model and question is motivated by many interesting problems from various, seemingly unrelated scientific areas. Some of them are: idea/rumor spreading, where the probability of spreading depends on the kind of idea/rumor; computer networks, where the probability that a message/malware will cover a set of terminals depends on the message/malware; and also spread of mutations, where the probability of a mutation occupying the whole population of cells depends on the mutation. 
	
	\noindent \textbf{A motivating example.} Using the latter application as an analogue for the rest, we give the following example to elaborate on the natural meaning of this process. Imagine a population of identical somatic \textit{resident} cells (e.g. biological tissue) that carry out a specific function (e.g. an organ). The cells connect with each other in a certain way; i.e., when a cell reproduces it replaces another from a specified set of candidates, that is, the set of cells connected to it. Reproduction here is the replication of the genetic code to the descendant, i.e. the hardwired commands which determine how well the cell will adapt to its environment, what its chances of reproduction are and which candidate cells it will be able to reproduce onto. 
	
	The changes in the information carried by the genetic code, i.e. mutations, give or take away survival or reproductive abilities. A bad case of mutation is a cancer cell whose genes force it to reproduce relentlessly, whereas a good one could be a cell with enhanced functionality. A mutation can affect the cell's ability to adapt to the environment, which translates to chances of reproduction, or/and change the set of candidates in the population that should pay the price for its reproduction.
	
	Returning to our population of resident cells, suppose that, after lots of reproductions a mutant shows up due to replication mistakes, environmental conditions, etc. This \textit{mutant} has the ability to reproduce in a different rate, and also, to be connected with a set of cells different than the one of its resident version. For the sake of argument, we study the most pessimistic case, i.e. our mutant is an extremely aggressive type of cancer with increased reproduction rate and maximum unpredictability; it can replicate on any other cell and do that faster than a resident cell. We consider the following motivating question: Supposing this single mutant will appear at some point in time on a random cell equiprobably, what is the best structure (network) of our resident cells such that the probability of the mutant taking over the whole population is minimized?
	
	Numerous studies in Biology model the dynamics of cancer progression in somatic cells using the Moran process \cite{IMN04,Kom06,MIN04,KLVN02,NMI03}. A mutation that affects the aforementioned characteristics in a real population of somatic cells occurs rarely compared to the time it needs to conquer the population or get extinct. Therefore, a second mutation is extremely rare to happen before the first one has reached one of those two outcomes and this allows us to study only one type of mutant per process. In addition, apart from the different reproduction rate, a mutation can lead to a different ``expansionary policy'' of the cell, a property that to the authors' knowledge has not been studied so far.

	\section{Definitions}\label{definitions}
	Each of the population's individuals is represented by a label $i \in \{1,2,...,n\}$ and can have one of two possible types: $R$ (\textit{resident}) and $M$ (\textit{mutant}). We denote the \textit{set of nodes} by $V$, with $n=|V|$, and the \textit{set of resident (mutant) edges} by $E_R$ ($E_M$). The node connections are represented by directed edges; a node $i$ has a \textit{type R (M) directed edge} $(ij)_{R}$ ($(ij)_{M}$) towards node $j$ if and only if when $i$ is chosen and is of type $R$($M$) then it can reproduce onto $j$ with positive probability. The aforementioned components define two directed graphs; the \textit{resident graph} $G_{R} = (V,E_R)$ and the \textit{mutant graph} $G_{M} = (V,E_M)$. A node's type determines its fitness; residents have \textit{relative fitness} 1, while mutants have relative fitness $r>0$. 
	
	The process works as follows: we start with every node of the vertex set being resident, except for one node which is selected uniformly at random to be mutant. We consider discrete time, and in each time-step some individual $i$ is picked with probability proportional to its fitness, and copies its resident/mutant type onto an out-neighbour $j$ in the corresponding graph ($G_{R}$ or $G_{M}$, depending on whether $i$ is resident or mutant) with probability determined by the weight of the edge $(ij)$ (of the corresponding graph). Given that the resident (mutant) $i$ is chosen for reproduction, the probability of reproducing on $j$ is equal to some \textit{weight} $w_{ij}^{R}$ ($w_{ij}^{M}$), thus $\sum_{j=1}^{n}w_{ij}^{R}=\sum_{j=1}^{n}w_{ij}^{M}=1$ for every $i \in V$. For $G_{R}$, every edge $(ij)_{R}$ has weight $w_{ij}^{R}>0$ if $(ij)_{R} \in E_{R}$, and $w_{ij}^{R}=0$ otherwise. Similarly for $G_{M}$. For each graph we then define \textit{weight matrices} $W_R = \left[w_{ij}^{R}\right]$ and $W_M = \left[w_{ij}^{M}\right]$ which contain all the information of the two graphs' structure. An example is shown in Figure \ref{example1}. After each time-step three outcomes can occur: (i) a node is added to the \textit{mutant set} $S \subseteq V$, (ii) a node is deleted from $S$, or (iii) $S$ remains the same. If both graphs are strongly connected the process ends with probability 1 when either $S=\emptyset$ (\textit{extinction}) or $S=V$ (\textit{fixation}).
	\begin{figure}[h]
		\begin{center}
			\begin{minipage}[h]{.22\textwidth}
				\includegraphics[scale=0.7]{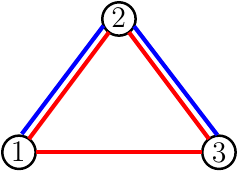}
			\end{minipage}
			\begin{minipage}[h]{.60\textwidth}
				\[
				\qquad W_R=
				\begin{bmatrix}
				0 					& 1 				 		& 0 							\\
				\frac{1}{2}	& 0							& \frac{1}{2}			\\
				0						& 1							& 0
				\end{bmatrix} \qquad
				W_M=
				\begin{bmatrix}
				0 					& \frac{1}{2} 	& \frac{1}{2} 		\\
				\frac{1}{2}	& 0							& \frac{1}{2}			\\
				\frac{1}{2}	& \frac{1}{2}		& 0
				\end{bmatrix}
				\]
			\end{minipage}
		\end{center}
		\caption{The two graphs combined; the edges of the resident graph are blue and the edges of the mutant graph are red. The respective weight matrices capture all the structure's information, including the weights to each edge. For example, the resident behaviour for node 1 (if chosen) is to reproduce only onto node 2, while its mutant behaviour is to reproduce equiprobably onto either 2 or 3.}\label{example1}
	\end{figure}

	We denote by $f(S)$ the probability of fixation given that we start with the mutant set $S$. We define the \textit{fixation probability} to be $f=\frac{1}{n}\sum_{u \in V}f\left(\{u\}\right)$ for a fixed relative fitness $r$. Notice the distinction between the terms ``probability of fixation given mutant set $S$'', and ``fixation probability''. We also define the \textit{extinction probability} to be equal to $1-f$. In the case of only one graph $G$ (i.e. $G_R=G_M=G$), which has been the standard setting so far in the literature, the point of reference for a graph's behaviour is the fixation probability of the complete graph (called \textit{the Moran fixation probability}) $f_{\text{Moran}}=\left(1-\frac{1}{r}\right)/\left(1-\frac{1}{r^n}\right)$. $G$ is an \textit{amplifier of selection} if $f>f_\text{Moran}$ and $r>1$ or $f<f_\text{Moran}$ and $r<1$ because it favors advantageous mutants and discourages disadvantageous ones. $G$ is a \textit{suppressor of selection} if $f<f_\text{Moran}$ and $r>1$ or $f>f_\text{Moran}$ and $r<1$ because it discourages advantageous mutants and favors disadvantageous ones.
	
	An \textit{unweighted graph} is a graph with the property that for every $i \in V$, $w_{ij}=\frac{1}{deg(i)}$ for every $j$ with incoming edge from $i$, where $deg(i)$ is the out-degree of node $i$. An undirected graph is a graph $G$ for which $(ij) \in E$ if and only if $(ji) \in E$, and therefore for any $i \in V$, the in-degree of $i$ equals its out-degree. In this work we will only deal with graphs that are undirected unweighted, and in the sequel we will abuse the term \textit{undirected graph} to refer to an undirected unweighted graph. Note that for an undirected graph, it could be $w_{ij} \neq w_{ji}$.
	
%	In what follows we will use special names to refer to some specific graph families. The following graphs have $n$ nodes which we omit from the notation for simplicity.
%	\begin{itemize}
%		\item $CL$ as a shorthand for the Clique or complete graph $K_n$.
%		\item $UST$ as a shorthand for the Undirected Star graph $K_{1,n-1}$.
%%		\item $UCY$ as a shorthand for the Undirected Cycle.
%		\item $CId$: as a shorthand for the Circulant graph $Ci_{n}(1,2,\dots,d/2)$ for even $d$. Briefly this subclass of circulant graphs is defined as follows: consider the cycle graph, and add edges between all pairs of nodes at distance at most $d/2$ from each other.
%	\end{itemize}

	Throughout this work we will use resident and mutant graphs from some well known graph families, namely: 
	\begin{itemize}
		\item the complete graph $K_n$,
		\item the undirected star graph $K_{1,n-1}$, which we denote by $S_n$ for simplicity,
		\item the circulant graph $Ci_{n}(1,2,\dots,d/2)$ for even $d$, which we denote by $CId$. Briefly this subclass of circulant graphs is defined as follows: consider the cycle graph, and add edges between all pairs of nodes at distance at most $d/2$ from each other.
	\end{itemize}
	We always assume that our graphs are on $n$ nodes, and we omit $n$ from the notation for simplicity when this is possible.

	\begin{figure}
		\begin{center}
			\includegraphics[scale=0.45]{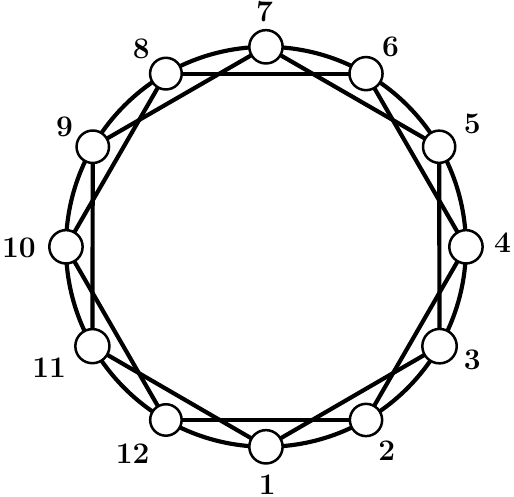}
			\includegraphics[scale=0.45]{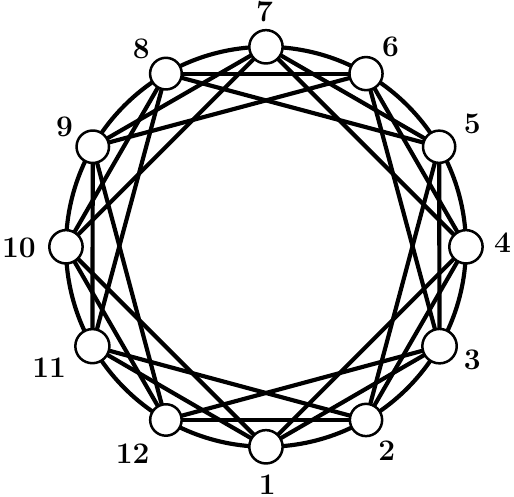}
			\includegraphics[scale=0.45]{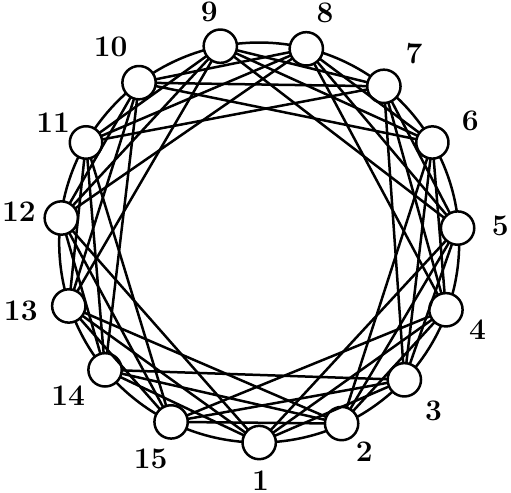}
			\caption{The classes of 4-regular, 6-regular and 8-regular undirected graphs $CI4$, $CI6$ and $CI8$. Here the number of nodes is $12, 12$ and $15$ respectively.}\label{CId}
		\end{center}
	\end{figure}
	By ``\textit{Resident Graph vs Mutant Graph}'' we refer to the process with $G_{R} =$ \textit{Resident Graph} and $G_{M} =$ \textit{Mutant Graph} and by $f(G_{R},G_{M} ; r)$ we refer to the fixation probability of that process with relative fitness $r$.
	Also, we denote by $f^{S}(G_R, G_M; r)$ the probability of fixation when our population has a set of mutants $S$ with relative fitness $r>0$, resident graph $G_R$ and mutant graph $G_M$. For presentation purposes we will omit $r$ as an argument of $f$ and $f^{S}$ whenever what is written is clear from context. Some of the arguments of the probability of fixation might be omitted when clear from context.  %Furthermore, $e_{G_{R},G_{B}} = 1 - f_{G_{R},G_{B}}$ is the extinction probability.
	
	We note that in this work, we are interested in the asymptotic behavior of the fixation probability as the population size $n$ grows. Therefore, we employ the standard asymptotic notation with respect to $n$; in particular, $r$ is almost always treated as a constant. Furthermore, in the rest of this work, by $G_R$ and $G_M$ we mean graph families $\{(G_{R})_{n}\}_{n \geq 5}$ and $\{(G_{M})_{n}\}_{n \geq 5}$ respectively, and we might omit $n$ when we only care about the fixation probability as $n \to \infty$.

	\section{Our Results}
	
	In this work, we introduce and study a generalization of the model of \cite{L05} by assuming that different types of individuals perceive the population through different graphs defined on the same vertex set, namely $G_R = (V, E_R)$ for residents and $G_M = (V, E_M)$ for mutants. In this model, we study the fixation probability, i.e. the probability that eventually only mutants remain in the population, for various pairs of graphs.
	
	In particular, in Section \ref{abstraction} we initially prove a tight upper bound (Theorem \ref{thm:generalupper}) on the fixation probability for the general case of an arbitrary pair of digraphs. Next, we prove a generalization of the Isothermal Theorem of \cite{L05}, that provides sufficient conditions for a pair of graphs to have fixation probability equal to the fixation probability of a clique pair, namely $f_\text{Moran} \stackrel{def}{=} f(K_n , K_n) = \left(1-\frac{1}{r}\right)/\left(1-\frac{1}{r^n}\right)$, for $r \in \mathbb{R}_{>0} \setminus \{1\}$; this corresponds to the absorption probability of a simple birth-death process with forward bias $r$. It is worth noting that it is easy to find small counterexamples of pairs of graphs for which at least one of the two conditions of Theorem 2 does not hold and yet the fixation probability is equal to $f_{\text{Moran}}$; hence we do not prove necessity.
	
	In Section \ref{strategicgame} we give a 2-player strategic game view of the process where player payoffs correspond to fixation and/or extinction probabilities. In this setting, we give an extensive study of the fixation probability when one of the two underlying graphs is complete, providing several insightful results. In particular, we prove that, the fixation probability $f(S_n , K_n) \in 1 - O(1/n)$, and thus tends to 1 as the number of nodes grows, for any constant $r>0$. By using a translation result (Lemma \ref{mirror}), we can show that, when the two graphs are exchanged, we have $f(K_n , S_n) \to 0$. Moreover, using a direct proof, in Theorem \ref{lemma5} we show that in fact $f(K_n , S_n) \in O\left(\frac{r^{n-1}}{(n-2)!}\right)$, i.e. it is exponentially small in $n$, for any constant $r>0$. In Theorem \ref{LB}, we also provide a lower bound on the fixation probability in the special case where the resident graph is any undirected graph and the mutant graph is a clique.
	
	Furthermore, in Subsection \ref{classesvsCL}, we find bounds on the fixation probability when the mutant graph is the clique and the resident graph belongs to large class of regular graphs, known as \textit{circulant graphs}. In particular, we show that when the mutant graph is the clique and the resident graph is a circulant graph $CId$ of any degree $d$, then $1-\frac{1}{r} - o(1) \leq f(CId, K_n) \leq e^{-1/r} + o(1)$, for any constant $r>2$, where $o(1)$ denotes a function of $n$ that tends to 0 as $n$ tends to $\infty$. Since the sparser circulant graph is the cycle graph ($CI2$), the latter result in fact reveals that such a simple graph structure for the residents performs very well against the mutant clique (for the aforementioned range of $r$). In addition, by running simulations (which we do not analyse here) for the case where the resident graph is the strong suppressor in \cite{GIA16}, and the mutant graph is the clique, we get fixation probability significantly greater than $f_\text{Moran}$ for up to $336$ nodes and values of fitness $r>1$. All of our results seem to indicate that the clique is the most beneficial graph (in terms of player payoff in the game theoretic formulation). However, the proof of this conjecture remains an open problem.
	
	Finally, in Section \ref{sec:fpras} we consider the problem of efficiently approximating the fixation probability in our model. We point out that Theorem \ref{lemma5} implies that the fixation probability cannot be approximated via a method similar to \cite{D14}. However, when we restrict the mutant graph to be complete, we prove a polynomial (in $n$) upper bound for the absorption time of the generalized Moran process when $r > 2c \left( 1 + o(1)\right)$, where $c$ is the maximum ratio of degrees of adjacent nodes in the resident graph. The latter allows us to give a fully polynomial randomized approximation scheme (FPRAS) for the problem of computing the fixation probability in this case.

	\section{Previous Work}
	So far the literature consists of works that consider the same structure for both residents and mutants. This single-graph setting was initiated by P.A.P. Moran \cite{M58} where the case of the complete graph was examined. Many years later, the setting was extended to structured populations on general directed graphs by Lieberman et al. \cite{L05}. They introduced the notions of amplifiers and suppressors of selection, a categorization of graphs based on the comparison of their fixation probabilities with that of the complete graph. They also found a sufficient condition (in fact \cite{G16} corrects the claim in \cite{L05} that the condition is also necessary) for a digraph to have the fixation probability of the complete graph, but a necessary condition is yet to be found.
	
	Since the generalized single-graph model in \cite{L05} was proposed, a great number of papers have tried to answer some very intriguing questions in this framework. One of them is the following: which are the best unweighted amplifiers and suppressors that exist? D\'iaz et al. \cite{D14} give the following bounds on the fixation probability of strongly connected digraphs: an upper bound of $1-\frac{1}{r+n}$ for $r>0$, a lower bound of $\frac{1}{n}$ for $r>1$ and they show that there is no positive polynomial lower bound when $0<r<1$. An interesting problem that was set in \cite{L05} is whether there are graph families that are \textit{strong amplifiers} or \textit{strong suppressors} of selection, i.e. families of graphs with fixation probability tending to 1 or to 0 respectively as the order of the graph tends to infinity and for $r>1$. Galanis et al. \cite{G16} presented an infinite family of strongly-amplifying directed graphs, namely the ``megastar'' with fixation probability $1-O(n^{-1/2}\log^{23}n)$, which was later proved to be optimal up to logarithmic factors \cite{GL16}.
	
	While the search for optimal directed strong amplifiers was still on, a restricted version of the problem had been drawing a lot of attention: which are the tight bounds on the fixation probability of undirected graphs? The lower bound in the undirected case remained $\frac{1}{n}$, but the upper bound was significantly improved by Mertzios et al. \cite{M13} to $1-\Omega(n^{-3/4})$, when $r$ is independent of $n$. It was again improved by Giakkoupis \cite{GIA16} to $1-\Omega\left(r^{-5/3} n^{-1/3}\log^{-4/3} n\right)$, and finally by Goldberg et al. \cite{GL16} to $1-\Omega(n^{-1/3})$ where they also found a graph which shows that this is tight. While the general belief was that there are no undirected strong suppressors, Giakkoupis \cite{GIA16} showed that there is a class of graphs with fixation probability $O(r^{2}n^{-1/4}\log n)$ and Goldberg et al. \cite{GLR18} found a stronger suppressor with fixation probability $O(r^{2} n^{-1/2})$, opening the way for a potentially optimal strong suppressor to be discovered.
	
	Extensions of \cite{L05} where the interaction between individuals includes a bimatrix game have also been studied. Ohtsuki et al. in\cite{OPN} considered the generalized Moran process with two distinct graphs, where one of them determines possible pairs that will play a bimatrix game and yield a total payoff for each individual, and the other determines which individual will be replaced by the process in each step. Two similar settings, where a bimatrix game determines the individuals' fitness, were studied by Ibsen-Jensen et al. in\cite{IJCN}. In that paper they prove NP-hardness and \#P-hardness on the computation of the fixation probabilities for the aforementioned settings respectively, and they also prove PSPACE inclusion for both. Regarding the problem of computing the fixation probability in the single-graph setting, \cite{CIN17} and \cite{GLR18} provide fully polynomial, randomized time approximation scemes (FPRAS) which have significantly improved running time compared to the previous algorithm of \cite{D14}.

	\section{Markov Chain Abstraction and the Generalized Isothermal Theorem}\label{abstraction}
	
	The \textit{two-graphs} Moran process that we propose is a generalization of the single-graph process. Like the single-graph process, this is also an absorbing Markov chain \cite{N06} whose states are the possible mutant sets $X \subseteq V$ ($2^n$ different mutant sets) and there are two absorbing states, namely $\emptyset$ and $ V $. 
	In this setting, the fixation probability is the average absorption probability to $V$, starting from a state with one mutant equiprobably. Since our Markov chain contains only two absorbing states, the sum of the fixation and extinction probabilities is equal to 1.
	
	~\\

	\noindent \textbf{Transition probabilities and domination arguments.}
	In the sequel, we will denote by $X+y$ the set $X \cup \{y\}$ and by $X-y$ the set $X \setminus \{y\}$. We can easily deduce the boundary conditions from the definition: $f(\emptyset)=0$ and $f(V)=1$. For any other arbitrary state $X$ of the process we have:
	\begin{align}\label{fixation1}
	f(X)= \sum_{i \in X, j \notin X}\frac{r}{F(X)}w_{ij}^{M}\cdot f(X+j) + \sum_{j \notin X, i \in X}\frac{1}{F(X)}w_{ji}^{R}\cdot &f(X-i) + \nonumber \\ 
	+ \left(\sum_{i \in X, j \in X}\frac{r}{F(X)}w_{ij}^{M} + \sum_{i \notin X, j \notin X}\frac{1}{F(X)}w_{ij}^{R}\right)\cdot &f(X) ,
	\end{align}
	where  $F(X)=|X|r+|V|-|X|$  is the total fitness of the population in state $ X $. By eliminating self-loops, we get
	\begin{align}\label{fixat2}
	f(X)= \frac{r \cdot \sum_{i \in X, j \notin X} w_{ij}^{M}}{r \cdot \sum_{i \in X, j \notin X} w_{ij}^{M} + \sum_{j \notin X, i \in X}w_{ji}^{R}}\cdot f(X+j) + \frac{\sum_{j \notin X, i \in X}w_{ji}^{R}}{r \cdot \sum_{i \in X, j \notin X} w_{ij}^{M} + \sum_{j \notin X, i \in X}w_{ji}^{R}}\cdot f(X-i) . 
	\end{align}

	We should note here that, in the general case, the fixation probability can be computed by solving a system of $2^n$ linear equations using this latter relation. However, bounds are usually easier to be found and special cases of resident and mutant graphs may have efficient exact solutions.
	
	In the proofs of this work we shall use the below fact from Example 3.9.6 in \cite{GS01} and also the results that follow it:
	\begin{fact}\label{remark2}
		Consider the birth-death process with state space $\{0, 1, \ldots, n\}$ and absorbing states $0, n$. Let the transition probability from state $a$ to state $b$ be denoted by $p_{a}^b$ and the backward bias at state $k \in \{1,2, \dots, n-1\}$ be the quantity $\gamma_k := p_{k}^{k-1}/p_{k}^{k+1}$. Then the probability of absorption at $n$, given that we start at $i$ is 
		\begin{align*}
		f_i = \frac{1+\sum_{j=1}^{i-1}\prod_{k=1}^{j}\gamma_k}{1+\sum_{j=1}^{n-1}\prod_{k=1}^{j}\gamma_k} .
		\end{align*} 
	\end{fact}
	
	For the next results which will be used to prove domination arguments among Moran processes throughout this work, we will need the following definitions. Consider the Moran process (discrete-time Markov chain) $(S(t))_{t\geq0}$ on two graphs $G_{R} = (V, E_R)$, $G_{M} = (V, E_M)$, where $|V| = n$, $S(t) \in 2^V$ and, by definition, its absorbing states are $\emptyset, V$. 
	
	The \textit{transition probability} from mutant set $X$ to mutant set $Y$ is $p_{X}^{Y} = \Pr \{ S(t+1)=Y \text{ } | \text{ } S(t)=X \}$. These probabilities define a $2^n \times 2^n$ \textit{transition matrix} $\mathbf{p}$ for which $\mathbf{p}_{X,Y} = p_{X}^Y$.\footnote{To enumerate the $2^n$ mutant sets, one simply can uniquely map each mutant set $X$ to its binary representation. E.g. by fixing an arbitrary order of the nodes, let the $i$-th bit correspond to node $i$ and give it value 1 if it is mutant in state $X$, and 0 otherwise. When we write $\mathbf{p}_{X,Y}$ we mean the element of $\mathbf{p}$ that is in the $X$-th (in binary representation) row and $Y$-th column.} 
%	Also, consider the partition $A_1, A_2, \dots$ of the sample space. The \textit{transition probability conditional on the event} $A \in \{ A_1, A_2, \dots \}$ is $p_{X}^Y (A) = \Pr \{ S(t+1) = Y \text{ } | \text{ } (S(t) = X) \cap A \}$.

\noindent	
	\textit{\underline{Notation remark:}}
		\textit{If for some $i \in \{1,2, \dots, n-1\}$ the transition probabilities $p_{X}^{Y}$ are the same for every $X,Y$ for which $|X| = i$, $|Y| = i+1$, we abuse the notation and denote all of them by $p_{i}^{i+1}$. Similarly, when $|X| = i$, $|Y| = i-1$ we denote all of these transition probabilities by $p_{i}^{i-1}$. 
%		When the aforementioned probabilities are conditional on some event $A$, they are denoted by $p_{i}^{i+1}(A)$ and $p_{i}^{i-1}(A)$, respectively.
	}

%	If the transition probability depends only on the cardinality of the current and next step, then we call it \textit{simple transition probability} and for ease of presentation we denote it by $p_{a}^b$. That is, if for any fixed pair $a,b \in \{ 0, 1, \dots, n \}$ the probability $\Pr \{ |S(t+1)|=b \text{ } | \text{ } |S(t)|=a \}$ is the same for every $t \geq 0$ then $ p_{a}^b = \Pr \{ |S(t+1)|=b \text{ } | \text{ } |S(t)|=a \} $. These define a $(n+1) \times (n+1)$ \textit{simple transition matrix} $\mathbf{p}$ for which $\mathbf{p}_{a,b} = p_{a}^b$. Similarly to the general case, we define the \textit{simple transition probability conditional on the event} $A \in \{ A_1, A_2, \dots \}$ to be $p_{a}^b (A) = \Pr \{ |S(t+1)|=b \text{ } | \text{ } |S(t)|=a \cap A \}$.

%%%%%%%%%%%%%%%%%%%%%%%%%%%%%%%%%%%%%%%%

As it is apparent by the Markov chain abstraction, the fixation probability $f(X)$ starting from a mutant set $X$ is the absorption probability for the absorbing state $V$ \cite{N98}. This absorption probability can be calculated from the set of equations \eqref{fixation1}, or after eliminating self-loops, from the equations \eqref{fixat2}. For ease of presentation, let us focus on the latter, and also denote by $p_{X}^{X+j}$ and $p_{X}^{X-i}$ the transition probabilities which are coefficients of $f(X+j)$ and $f(X-i)$ in \eqref{fixat2}.
	
Let $M=(M(t))_{t \geq 0}$ be a Moran process, and for each $k' \in \{0,1,\dots, n\}$ consider the set $S_{k'} := \{S \in 2^{V} ~ | ~ |S|=k'\}$ consisting of states with $k'$ mutants. Note that $S_0 = \{\emptyset\}$ and $S_n = \{V\}$. 
For any $k \in \{1, 2, \dots, n-1\}$ and any state $X \in S_k$, let us define the \textit{cumulative backward bias} of $X$ in process $M$ as $\gamma^{c}(X;M) := \frac{\sum_{j \in S_{k-1}} p_{X}^j}{\sum_{j \in S_{k+1}} p_{X}^j}$. We will omit the initial state or/and process indicators when clear from context. We also define $S^{(k)}_{\max} := \arg \max_{X \in S_k} \gamma^{c}(X)$ and $S^{(k)}_{\min} := \arg \min_{X \in S_k} \gamma^{c}(X)$. Now, let $\gamma'_{k}, \gamma''_{k} > 0$ be such that $\gamma'_{k} \geq \max_{X \in S_k} \gamma^{c}(X)$ and $\gamma''_{k} \leq \min_{X \in S_k} \gamma^{c}(X)$. Also, let the Moran processes $M_{\max}$ and $M_{\min}$ be such that the only state with $k$ mutants is $S^{(k)}_{\max}$ and $S^{(k)}_{\min}$, respectively, for each $k \in \{1, 2, \dots, n-1\}$. Their respective fixation probabilities starting from state $S \subseteq V$ are denoted by $f_{\max}(S;M)$ and $f_{\min}(S;M)$, and again the indicators $S$ or/and $M$ are omitted when clear from context. The transition probabilities of these processes are $p_{S^{(k)}_{\max}}^{S^{(k+1)}_{\max}} := \sum_{j \in S_{k+1}} p_{S^{(k)}_{\max}}^j$, $p_{S^{(k)}_{\max}}^{S^{(k-1)}_{\max}} := \sum_{j \in S_{k-1}} p_{S^{(k)}_{\max}}^j$, and $p_{S^{(k)}_{\min}}^{S^{(k+1)}_{\min}} := \sum_{j \in S_{k+1}} p_{S^{(k)}_{\min}}^j$, $p_{S^{(k)}_{\min}}^{S^{(k-1)}_{\min}} := \sum_{j \in S_{k-1}} p_{S^{(k)}_{\min}}^j$, respectively. 

Finally, consider the fixation probability of Fact \ref{remark2}, i.e. a function that depends on $i \in \{1, 2, \dots, n-1 \}$ and $\widetilde{\gamma} := (\gamma_{1}, \gamma_{2},\dots, \gamma_{n-1}) \in \mathbb{R}_{>0}^{n-1}$, and let us denote it by
\begin{align*}
	g_{i}(\widetilde{\gamma}) := \frac{1+\sum_{j=1}^{i-1}\prod_{k=1}^{j}\gamma_k}{1+\sum_{j=1}^{n-1}\prod_{k=1}^{j}\gamma_k} .
\end{align*} 
%
%
%\begin{lem}
%	Let $S \subseteq V$ be an initial state with $i \in \{0,1,\dots, n-1\}$ mutants in Moran process $M$, and $f(S)$ be the fixation probability. Then $g_{i}(\widetilde{\gamma'}) \leq f_{\max}(S) \leq f(S) \leq f_{\min}(S) \leq g_{i}(\widetilde{\gamma''})$ , where $f_{\max}, f_{\min}$ are the fixation probabilities of $M_{\max}$ and $M_{\min}$ respectively, and $\widetilde{\gamma'} = (\gamma'_{1}, \gamma'_{2},\dots, \gamma'_{n-1})$, $\widetilde{\gamma''} = (\gamma''_{1}, \gamma''_{2},\dots, \gamma''_{n-1})$.
%\end{lem}

\begin{lem} \label{lem: domination}
	Let $S \subseteq V$ be an initial state with $i \in \{0,1,\dots, n-1\}$ mutants in Moran process $M$, and $f(S)$ be the fixation probability. Then 
	\begin{align*}
		g_{i}(\widetilde{\gamma'}) \leq f_{\max}(S) \leq f(S) \leq f_{\min}(S) \leq g_{i}(\widetilde{\gamma''}) , 
	\end{align*}
	where $\widetilde{\gamma'} = (\gamma'_{1}, \gamma'_{2},\dots, \gamma'_{n-1})$, $\widetilde{\gamma''} = (\gamma''_{1}, \gamma''_{2},\dots, \gamma''_{n-1})$.
\end{lem}

\begin{proof}
We will present the proof for the upper bounds of $f(S)$; for the lower bounds the proof is symmetric, that is why we omit it. Also, we will consider an initial state $S \in S_1$ but, as discussed towards the end, the proof goes through for $S \in S_{i}$ for any $i \in \{1, 2, \dots, n-1\}$. The fixation probability $f(S)$ can be found by solving the following set of equations (see also equation \eqref{fixat2}):
\begin{align}\label{fixat_system}
	f(S) &= \sum_{S^{(2)} \in S_{2}} p_{S}^{S^{(2)}} f(S^{(2)})   \nonumber  \\
	f(S^{(1)}) &= \sum_{S^{(2)} \in S_{2}} p_{S^{(1)}}^{S^{(2)}} f(S^{(2)})   , \quad \text{for every } S^{(1)} \in S_1 \setminus \{S\}   \nonumber  \\
	f(S^{(2)}) &= \sum_{S^{(3)} \in S_{3}} p_{S^{(2)}}^{S^{(3)}} f(S^{(3)}) + \sum_{S^{(1)} \in S_{1}} p_{S^{(2)}}^{S^{(1)}} f(S^{(1)}) , \quad \text{for every } S^{(2)} \in S_2     \nonumber   \\
	&\vdots \\
	f(S^{(n-2)}) &= \sum_{S^{(n-1)} \in S_{n-1}} p_{S^{(n-2)}}^{S^{(n-1)}} f(S^{(n-1)}) + \sum_{S^{(n-3)} \in S_{n-3}} p_{S^{(n-2)}}^{S^{(n-3)}} f(S^{(n-3)}) , \quad \text{for every } S^{(n-2)} \in S_{n-2}    \nonumber    \\
	f(S^{(n-1)}) &= p_{S^{(n-1)}}^{V} + \sum_{S^{(n-2)} \in S_{n-2}} p_{S^{(n-1)}}^{S^{(n-2)}} f(S^{(n-2)}) , \quad \text{for every } S^{(n-1)} \in S_{n-1}   .     \nonumber  
%	
%	f(S) &= \sum_{S' \in S_{k+1}} p_{S}^{S'} f(S') + \sum_{S'' \in S_{k-1}} p_{S}^{S''} f(S'')
\end{align}

We will show that the fixation probability $f(S)$ is upper bounded by a sequence of fixation probabilities of modified Moran processes. In particular, let us denote by $M$ the initial Moran process, by $M_{1*}, M_{2*}, \dots, M_{n-1*}$ a sequence of modified processes, and by $f_{1*}(\cdot), f_{2*}(\cdot), \dots, f_{n-1*}(\cdot)$ the corresponding absorption probabilities as functions of the state space. We will show that there exists a state $S^{(1*)} \in S_1$ such that $f(S) \leq f_{1*}(S^{(1*)}) \leq f_{2*}(S^{(1*)}) \leq \dots \leq f_{n-1*}(S^{(1*)})$.

Consider the initial Moran process $M$ and the state $S^{(1*)} := \arg \max_{S^{(1)} \in S_1} f(S^{(1)})$ i.e. the one that has the maximum fixation probability in $S_1$. By definition, we have $f(S) \leq f(S^{(1*)})$.
We define the modified process $M_{1*}$ to be the process in which for any given $S^{(2)} \in S_2$, all the directed edges towards states in $S_1$ (i.e. having probabilities $p_{S^{(2)}}^j$, $j \in S_1$), are merged into one directed edge towards $S^{(1*)}$ which has probability $\sum_{j \in S_1} p_{S^{(2)}}^j$. Note that $M_{1*}$ has only $S^{(1*)}$ in $S_1$. Using the system of equations \eqref{fixat_system} for $M_{1*}$ and comparing them with those of $M$, it is immediate that $f(S^{(1*)}) \leq f_{1*}(S^{(1*)})$. 

We will now inductively define all modified processes. Given the Moran process $M_{k-1*}$, $k \in \{2, 3, \dots, n-1\}$, we define $M_{k*}$ to be the modified $M_{k-1*}$ where the only state in $S_k$ is $S^{(k*)} := \arg \max_{S^{(k)} \in S_k} f_{k-1*}(S^{(k)})$. In particular, for any given $S^{(k+1)} \in S_{k+1}$, all the directed edges towards states in $S_k$ (i.e. having probabilities $p_{S^{(k+1)}}^j$, $j \in S_k$), are merged into one directed edge towards $S^{(k*)}$ which has probability $\sum_{j \in S_k} p_{S^{(k+1)}}^j$. Also, for any given $S^{(k-1)} \in S_{k-1}$, all the directed edges towards states in $S_k$ (i.e. having probabilities $p_{S^{(k-1)}}^j$, $j \in S_k$), are merged into one directed edge towards $S^{(k*)}$ which has probability $\sum_{j \in S_k} p_{S^{(k-1)}}^j$. Again, note that $M_{k*}$ has only $S^{(k*)}$ in $S_k$, and that by using the system of equations \eqref{fixat_system} for it we get $f_{k-1*}(S^{(1*)}) \leq f_{k*}(S^{(1*)})$.

From the above, it is implied that $f(S) \leq f_{n-1*}(S^{(1*)})$, which allows us to draw two very important conclusions. First, that $M_{n-1*}$ is a Moran process with only $n+1$ states $\emptyset, S^{(1*)}, S^{(2*)}, \dots, S^{(n-1*)}, V$, and therefore can be also viewed as a random walk on these states. Second, one can observe that starting from any state $S \in S_i$, for any $i \in \{1,2, \dots, n-1\}$ and using an almost identical proof, we can upper bound the fixation probability of any other non-singleton mutant set. 

The aforementioned conclusions immediately imply that given any $S \in S_i$, for any $i \in \{1,2, \dots, n-1\}$, the upper bound $f_{n-1*}(S^{(i*)})$ can be computed by the formula of Fact \ref{remark2}. The fact that the final process $M_{n-1*}$ has only $n+1$ states is positive, since its fixation probabilities can be computed by solving only that many linear equations. However, observe that we do not have a polynomial time algorithm to do that, since to find the coefficients (sums of transition probabilities) of the equations of $M_{n-1*}$ we need to find the special states $S^{(k*)} \in S_k$ for each $k \in \{1,2, \dots, n-1\}$, and we do not have an efficient algorithm for that (some $S_k$'s have exponentially many states).

There is a way to circumvent this obstacle though, by further upper bounding $f_{n-1*}(S^{(k*)})$. Note that for the formula of Fact \ref{remark2}, we have
\begin{align*}
	f_{n-1*}(S^{(i*)}) = \frac{1+\sum_{j=1}^{i-1}\prod_{k=1}^{j}\gamma_k}{1+\sum_{j=1}^{n-1}\prod_{k=1}^{j}\gamma_k} = \frac{1}{1+ \frac{\sum_{j=i}^{n-1}\prod_{k=1}^{j}\gamma_k}{1+\sum_{j=1}^{i-1}\prod_{k=1}^{j}\gamma_k}} =  \frac{1}{1+ \frac{1}{1 + \sum_{j=1}^{i-1} \left(\prod_{k=1}^{j}\gamma_k\right)^{-1}} \cdot \sum_{j=i}^{n-1} \prod_{k=i}^{j}\gamma_k},
\end{align*} 
where $\gamma_{k}$ is the backward bias of state $S^{(k*)}$. It is again immediate that this is upper bounded by replacing each $\gamma_{k}$ with some $\gamma^{*}_{k} \leq \gamma_{k}$. Therefore, to (efficiently) compute an upper bound on $f_{n-1*}(S^{(i*)})$ it suffices that we (efficiently) compute a lower bound on the cumulative backward biases in each $S_k$, $k \in \{1,2, \dots, n-1\}$. When for each $k$ we choose $\gamma^{*}_{k} = \min_{X \in S_k} \gamma^{c}(X)$, we get fixation probability $f_{\min}(S)$, whereas if we choose some $\gamma^{*}_{k} \leq \min_{X \in S_k} \gamma^{c}(X)$ we get fixation probability $g_{i}(\widetilde{\gamma^*}) \geq f_{\min}(S)$. This completes the proof for the upper bounds.

It is easy to see that a symmetric proof holds for lower bounding a fixation probability of any given Moran process $M$ for any mutant set $S \subseteq V$: states $S^{(k*)}$ are defined as minimizers of fixation, and $f_{n-1*}(S^{(i*)})$ is further lower bounded by considering backward biases $\gamma^{*}_{k} \geq \gamma_{k}$. 

\end{proof}

In most of the results of this work we will utilize the above lemma and the following corollary that comes directly from it. In particular, by giving bounds on the backward biases of the process at hand we will be bounding their fixation probabilities.

\begin{corollary} \label{cor:domination}
	Consider two Moran processes $P=(P(t))_{t\geq0}$ and $Q=(Q(t))_{t\geq0}$ with transition matrices $\mathbf{p}$ and $\mathbf{q}$ respectively. Let the initial state be $S \subseteq V$. If for any given $X,Y \in 2^V$ it holds that $p_{X}^{Y} \leq q_{X}^{Y}$ whenever $|Y|=|X|-1$, and $p_{X}^{Y} \geq q_{X}^{Y}$ whenever $|Y|=|X|+1$, then $f(S;P) \geq f_{\max}(S;Q)$, and $f(S;Q) \leq f_{\min}(S;P)$.
\end{corollary}
	
\begin{proof}
Assume that the conditions of the corollary's statement are true. Then, for any given $X,Y,Y' \in 2^V$ for which it holds that $|Y|=|X|-1$ and $|Y'|=|X|+1$, we have $p_{X}^{Y} \leq q_{X}^{Y}$, and $p_{X}^{Y'} \geq q_{X}^{Y'}$. 
Consequently, 
%it holds that $p_{X}^{Y}/p_{X}^{Y'} \leq q_{X}^{Y}/q_{X}^{Y'}$, and therefore, 
for the cumulative backward biases of $X$ in processes $P$ and $Q$ it holds that
\begin{align} \label{eq: cum_back_bias_ineq}
	\gamma^{c}(X;P) \leq \gamma^{c}(X;Q) , \quad \text{for every $X \in 2^V \setminus \{\emptyset, V\}$.} 
\end{align}

First, from \eqref{eq: cum_back_bias_ineq} we get $\max_{X \in S_k} \gamma^{c}(X;P) \leq \max_{X \in S_k} \gamma^{c}(X;Q)$ for every $k \in \{1,2,\dots, n-1\}$. So, by Lemma \ref{lem: domination} for $\gamma'_{k} = \max_{X \in S_k} \gamma^{c}(X;Q)$ for all $k \in \{1,2,\dots, n-1\}$, we get $f(S;P) \geq f_{\max}(S;P) \geq f_{\max}(S;Q)$. 

Also, again from \eqref{eq: cum_back_bias_ineq} we get $\min_{X \in S_k} \gamma^{c}(X;Q) \geq \min_{X \in S_k} \gamma^{c}(X;P)$ for every $k \in \{1,2,\dots, n-1\}$. Therefore, by Lemma \ref{lem: domination} for $\gamma''_{k} = \min_{X \in S_k} \gamma^{c}(X;P)$ for all $k \in \{1,2,\dots, n-1\}$, we get $f(S;Q) \leq f_{\min}(S;Q) \leq f_{\min}(S;P)$. 
\end{proof}

	Also, for some of the proofs in this work we will need to employ additional stochastic domination results. The following set of results are due to \cite{DGRS16}, and rigorously establish a fact that had been intuitively accepted and used in previous works: removing a mutant from the mutant set and/or lowering their relative fitness can only decrease the fixation probability. In \cite{DGRS16} the results are for the single-graph Moran process on any digraph; however we find that their proofs hold for the two-graphs Moran process as well. 
	
	In the aforementioned paper, first, a coupling between two continuous-time Moran processes on an arbitrary digraph is provided. The continuity of the processes gets around a disturbing but yet interesting problem: a coupling between discrete Moran processes is not always possible. The intuition for the root of this problem is that when a node $v$ becomes mutant (with fitness $r > 1$), all other nodes become less likely to reproduce in the next round, and additionally, all nodes have to co-ordinate and point to a single node to reproduce in the next step. The coupling is achieved by continuous versions of the Moran processes, where, at any time $t$ each node $v$ reproduces at $t+\tau_v$, with $\tau_v$ being chosen according to the exponential distribution. The distribution $v$ is independent of any other node's distribution and parameterized by $v$'s fitness. 
	The discrete-time Moran process is then recovered by considering the sequence of reproductions that occurred throughout the continuous-time process.
	
	The following Lemma is a generalized statement of that of Lemma 5 in \cite{DGRS16}, since it captures two-graphs Moran processes. Its proof is omitted due to the fact that it is identical to the one regarding single-graph Moran processes.
	
	\begin{lem}[Coupling lemma - Lemma 5, \cite{DGRS16}]
		Let $G_{R}=(V,E_R)$ and $G_{M}=(V,E_M)$ be any two digraphs, let $Y \subseteq Y' \subseteq V$ and $1 \leq r \leq r'$. Let $\widetilde{Y}(t)$ and $\widetilde{Y}'(t)$ ($t \geq 0$) be continuous-time Moran processes on the pair ($G_R, G_M$) with mutant fitness $r$ and $r'$, respectively, and with $\widetilde{Y}(0) = Y$ and $\widetilde{Y}'(0) = Y'$. There is a coupling between the two processes such that $\widetilde{Y}(t) \subseteq \widetilde{Y}'(t)$ for all $t \geq 0$.
	\end{lem}
	
	The theorem below together with its two following corollaries are also generalized statements of those of Theorem 6 and Corollaries 7 and 8 in \cite{DGRS16}. All of these results regarding the single-graph Moran process hold for the two-graphs Moran process since their proofs are unaffected by the generalization. 
	
	\begin{thm}[Theorem 6, \cite{DGRS16}]
		For any pair of digraphs $G_R$ and $G_M$, if $0 < r \leq r'$ and $S \subseteq S' \subseteq V$, then $f^{S}(G_R, G_M ; r) \leq f^{S'}(G_R, G_M ; r')$
	\end{thm}
	
	\begin{corollary}[Monotonicity - Corollary 7, \cite{DGRS16}]
		If $0 < r \leq r'$ then, for any pair of digraphs $G_R$, $G_M$, we have $f(G_R, G_M ; r) \leq f(G_R, G_M ; r')$.
	\end{corollary}
	
	\begin{corollary}[Subset domination - Corollary 8, \cite{DGRS16}]\label{cor:subset_dom}
		For any pair of digraphs $G_R$, $G_M$ and any $r > 0$, if $S \subseteq S' \subseteq V$, then $f^{S}(G_R, G_M ; r) \leq f^{S'}(G_R, G_M ; r)$.
	\end{corollary}
	
	%%%%%%%%%%%%%%%%%%%%%%%%%%%%%%%%%%%%%%%%%%

	Using the Markov chain abstraction from the begining of Section \ref{abstraction} and the latter stochastic domination arguments, we can prove the following general upper bound on the fixation probability:

	\begin{thm}\label{thm:generalupper}
		For any pair of digraphs $G_R$ and $G_M$ with $n=|V|$, the fixation probability $f(G_R, G_M)$ is upper bounded by $1 - \frac{1}{r+n}$, for $r>0$. For constant $r$, this bound is asymptotically tight.
	\end{thm}
	
	\begin{proof}
		We refer to the proof of Lemma 4 of \cite{D14}, as our proof is essentially the same. Briefly, we find an upper bound on the fixation probability of a relaxed Moran process that favors the mutants, where we assume that fixation is achieved when two mutants appear in the population. In their work the resident and mutant graphs are the same and undirected, but this does not change the probabilities of the first mutant placed uniformly at random to be extinct or replicated in our model. Finally, we note that this bound is asymptotically (as $n \to \infty$) tight when $r$ is constant, by Theorem \ref{lemma1}. 
	\end{proof}

	We now prove a generalization of the Isothermal Theorem of \cite{L05}.
	
	\begin{thm}[Generalized Isothermal Theorem]\label{GIT}
		Let $G_R=(V, E_R)$, $G_M=(V, E_M)$ be two directed graphs. The generalized Moran process with two graphs as described above has the Moran fixation probability if:
		\begin{enumerate}
			\item $\sum_{j \neq i}w_{ji}^R = \sum_{j \neq i}w_{ji}^M = 1$, $\forall i \in V$, that is, $W_R$ and $W_M$ are doubly stochastic, i.e. $G_R$ and $G_M$ are isothermal (actually one of them being isothermal is redundant as it follows from the second condition), and \\
			\item for every pair of nodes $i,j \in V$: $w_{ij}^R + w_{ji}^R = w_{ij}^M + w_{ji}^M$.
		\end{enumerate}
	\end{thm}

	In order to prove the above theorem, we first prove the following auxiliary lemma.
	\begin{lem}\label{lem:eq_con_GIT}
		The conditions of Theorem \ref{GIT} are equivalent to the following set of equalities:
		\begin{align}\label{rel1}
		\sum_{i \notin S}\sum_{j \in S}w_{ij}^R = \sum_{i \notin S}\sum_{j \in S}w_{ji}^M , \quad \forall \emptyset \subset S \subset V.
		\end{align}
	\end{lem}

	\begin{proof}
%		It suffices to show that in every state $S$ of the process (Markov chain) with $0<|S|<|V|$ mutants, the probability to go to a state with $|S|+1$ mutants is $r$ times the probability to go to a state with $|S|-1$ mutants (see Example 3.9.6 in \cite{GS01}). In our setting, by (\ref{fixation1}) these probabilities are $\left(r \cdot \sum_{i \in S}\sum_{j \notin S}w_{ij}^M \right)/F$ and $\left(\sum_{i \notin S}\sum_{j \in S}w_{ij}^R \right)/F$ respectively. So, to establish the theorem, it suffices to show that its hypotheses hold if and only if relation (\ref{rel1}) holds.
%		%
%		\begin{align}\label{rel1}
%		\sum_{i \notin S}\sum_{j \in S}w_{ij}^R = \sum_{i \notin S}\sum_{j \in S}w_{ji}^M , \quad \forall \emptyset \subset S \subset V.
%		\end{align}
%		
		We first show that equations \ref{rel1} are sufficient for Conditions 1 and 2 of Theorem \ref{GIT}.
		Consider all the states where only one node $i$ is resident, i.e. $S=V \setminus \{i\}$. Then from (\ref{rel1}) we get the following set of equations that must hold:
		\begin{align}\label{rel2}
		\sum_{j \in V \setminus \{i\}}w_{ji}^M = \sum_{j \in V \setminus \{i\}}w_{ij}^R = 1 , \quad \forall i \in V.
		\end{align}
		Similarly, for all the states where $S=\{i\}$ we get from (\ref{rel1}):
		\begin{align}\label{rel3}
		\sum_{j \in V \setminus \{i\}}w_{ji}^R = \sum_{j \in V \setminus \{i\}}w_{ij}^M = 1 , \quad \forall i \in V.
		\end{align}
		Observe that \eqref{rel2} and \eqref{rel3} constitute Condition 1 of Theorem \ref{GIT}.
		Now, (for general $S$) the two parts of (\ref{rel1}) are:
		\begin{align}
		\sum_{i \notin S}\sum_{j \in S}w_{ij}^R = |V|-|S| - \sum_{i \notin S}\sum_{j \notin S \cup \{i\}}w_{ij}^R  \label{rel4}
		\end{align}
		\begin{align}
		\text{and} \quad \sum_{i \notin S}\sum_{j \in S}w_{ji}^M = |V|-|S| - \sum_{i \notin S}\sum_{j \notin S \cup \{i\}}w_{ji}^M , \quad \text{(using (\ref{rel2})).}  \label{rel5}
		\end{align}
		Thus, by relation (\ref{rel1}) it must be: 
		\begin{align}
		\sum_{i \notin S}\sum_{j \notin S \cup \{i\}}w_{ij}^R = \sum_{i \notin S}\sum_{j \notin S \cup \{i\}}w_{ji}^M , \quad \forall \emptyset \subset S \subset V,.  \label{rel6}
		\end{align}
		Now, consider all the states where only two nodes $i$ and $j$ are resident, i.e. $S=V \setminus \{i,j\}$. Then from relation (\ref{rel6}) we get the following set of relations that must hold:
		\begin{align}\label{rel7}
		w_{ij}^R + w_{ji}^R = w_{ij}^M + w_{ji}^M  , \quad \forall i,j \in V,
		\end{align}
		which is Condition 2 of Theorem \ref{GIT}.
		
		To prove the other direction of the equivalence we show that the Conditions of Theorem \ref{GIT}, which correspond to the sets of relations \eqref{rel2}, \eqref{rel3}, \eqref{rel7}, suffice to make \eqref{rel1} true.
		If (\ref{rel7}) is true, then (\ref{rel6}) is obviously true. From (\ref{rel2}) and (\ref{rel3}), the right-hand sides of (\ref{rel4}) and (\ref{rel5}) are equal, therefore their left-hand sides are equal, and as a consequence (\ref{rel1}) is true. 
	\end{proof}

	\begin{proof}(\textit{of Theorem \ref{GIT}})
		We will prove that if the conditions of Theorem \ref{GIT} are satisfied, then the two-graphs Moran process has the Moran fixation probability, i.e. $f_{\text{Moran}}=\left(1-\frac{1}{r}\right)/\left(1-\frac{1}{r^n}\right)$ for any positive real $r \neq 1$.
		
		Suppose the given resident and mutant digraphs $G_R$ and $G_M$ satisfy the conditions of Theorem \ref{GIT}. For a given state of the Moran process with mutant set $S$, where $0 < |S| < |V|$, let $p_{|S|}^{|S|-1}$ and $p_{|S|}^{|S|+1}$ be the probabilities for the process to go to some state with $|S|-1$ and $|S|+1$ mutants, respectively. Using \eqref{fixation1} we have $p_{|S|}^{|S|-1} = \left(\sum_{i \notin S}\sum_{j \in S}w_{ij}^R \right)/F$ and $p_{|S|}^{|S|+1} = \left(r \cdot \sum_{i \in S}\sum_{j \notin S}w_{ij}^M \right)/F$. By Lemma \ref{lem:eq_con_GIT}, we have $p_{|S|}^{|S|-1} / p_{|S|}^{|S|+1} = 1/r$, for every $\emptyset \subset S \subset V$. Then the Moran process at hand can be modelled by a birth-death process as described in Fact \ref{remark2}, where additionally, the backward bias is $1/r$ independent of the state of the process (see Example 3.9.6 in \cite{GS01}). In Fact \ref{remark2}, a formula is provided for the fixation probability given that the mutant set $S$ consists of $i=|S|$ mutants. For $i=1$ the formula gives the fixation probability, which equals $\left(1-\frac{1}{r}\right)/\left(1-\frac{1}{r^n}\right)$. 
	\end{proof}
	
	Observe that for the special case $G_R=G_M$, Theorem \ref{GIT} is the Isothermal Theorem of the generalized Moran process of \cite{L05} that has been studied so far.

	\section{A Strategic Game View} \label{strategicgame}
	In this section we study the aforementioned process from a game-theoretic point of view. Consider the strategic game with 2 players; residents (type R) and mutants (type M), so the player set is $N=\{R,M\}$. The action set of a player $k \in N$ consists of all possible strongly connected graphs\footnote{We assume strong connectivity in order to avoid problematic cases where the fixation probability is 0.} $G_k = (V,E_k)$ that she can construct with the available vertex set $V$. The payoff for the residents (player $R$) is the probability of extinction, and the payoff for the mutants (player $M$) is the probability of fixation. Since fixation and extinction are complementary events, the sum of payoffs equals 1, so the game can be reduced to a zero-sum game.
	
	The natural question that emerges is: what are the pure Nash equilibria of this game (if any) when $r$ is fixed and $n \to \infty$? As a toy example consider some fixed constant $r > 1$, and that each player has the small action set $\{ K_n, S_n \}$. To avoid the case of both players choosing $S_n$ with different central nodes, for the sake of simplicity we add the constraint that the $S_n$ action for both players comes with some arbitrary, fixed $v \in V$ that is the centre. From our results from Section \ref{USTvsCL}, when $n \to \infty$, we get $f(K_n , S_n) \to 0, f(S_n , K_n) \to 1$ and from \cite{N06,BR08}, $f(K_n , K_n) \to 1-1/r$ and $f(S_n , S_n) \to 1 - 1/r^2$. Therefore, we get the following bimatrix game:

	\begin{table}[H]
		\setlength{\extrarowheight}{2pt}
		\begin{tabular}{cc|c|c|}
			& \multicolumn{1}{c}{} & \multicolumn{2}{c}{Player $M$}\\
			& \multicolumn{1}{c}{} & \multicolumn{1}{c}{$K_n$}  & \multicolumn{1}{c}{$S_n$} \\\cline{3-4}
			\multirow{2}*{Player $R$}  & $K_n$ & $1/r \quad,\quad 1-1/r$ & $1 \quad,\quad 0$ \\\cline{3-4}
			& $S_n$ & $0 \quad,\quad 1$ & $1/r^2 \quad,\quad 1-1/r^2$ \\\cline{3-4}
		\end{tabular}
	\end{table}
	\noindent which has a pure Nash equilibrium, namely $(K_n , K_n)$. 
	
	Trying to better understand the behaviour of the two conflicting families of graphs, we put some pairs of graph families to the test. The main question we ask in this work is: what is the best response graph $G_R$ of the residents to the Clique graph of the mutants?
	
	The bounds on the fixation or extinction probability in this work can be viewed as bounds on the payoff of player $M$ or player $R$ respectively.

	\subsection{Star vs Clique} \label{USTvsCL}
	
	The following result implies (since $(n-4)!^{-1/(n-2)} \to 0$ as $n \to \infty$) that when the mutant graph is complete and the resident graph is the undirected star, the fixation probability tends to 1 as $n \to \infty$.

	\begin{thm}\label{lemma1}
		If $G_R = S_n$ and $G_M = K_n$ for $r>(n-4)!^{-1/(n-2)}$, then the fixation probability is lower bounded by $1 - \frac{1}{n} - \frac{1}{r(n-2)} - \frac{1}{r^{2}(n-3)}$.
	\end{thm}
	
	\begin{proof}
		We will find a lower bound to the fixation probability of our process $P$, by giving a lower bound on the fixation probability of a process $Q$ that is dominated by (has at most the fixation probability of) $P$. 
	
		The Markov chains of $P$ and $Q$ are depicted in Figure \ref{UST}.
		\begin{figure}[H]
			\begin{center}
			\includegraphics[scale=0.68]{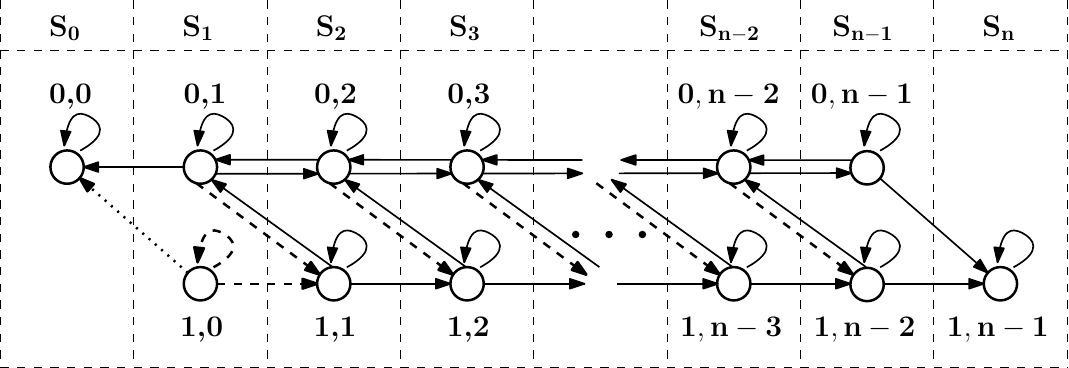}
			\caption{The Markov chains for process $P$ and $Q$. In $Q$, the dashed edges have probability $0$ (their former probability is added to the self loop) and the dotted one has probability $1$.}\label{UST}
			\end{center}
		\end{figure}
		We denote by $(c,\ell)$ the state of process $P$ that has $c$ mutants at the center of $G_R$ and $\ell$ mutants at the leaves of $G_R$. We have 
		\begin{align} \label{eq: f(P)}
			f(P) = \left(1-\frac{1}{n}\right) f((0,1);P) + \frac{1}{n} f((1,0);P) .
		\end{align}
	
		In $P$, the transition probabilities for $1 \leq i \leq n-1$ are:
		\begin{align*}
			p_{0,i}^{0,i-1} &= \frac{1}{ir+n-i} \cdot \frac{i}{n-1},  \quad p_{1,i-1}^{0,i-1} = \frac{n-i}{ir+n-i},  \\
			p_{0,i}^{0,i+1} &= \frac{ir}{ir+n-i} \cdot \frac{n-i-1}{n-1}, \quad		p_{0,i}^{1,i} = \frac{ir}{ir+n-i} \cdot \frac{1}{n-1},  \quad	p_{1,i-1}^{1,i} = \frac{ir}{ir+n-i} \cdot \frac{n-i}{n-1} .
		\end{align*}

		Let us define the intermediate process $Q'$ in which the only difference with $P$ is that $q_{1,0}^{0,0} = 1$, $q_{1,0}^{1,0} = q_{1,0}^{1,1} = 0$. Now observe that $f((1,0);Q') = 0$, and since by starting from $(0,1)$ we can never reach $(1,0)$ the systems of equations \eqref{fixat_system} for $P$ and $Q'$ are identical, thus we get 
		\begin{align} \label{eq: P,Q'}
			f((0,1);P) = f((0,1);Q').
		\end{align}

		Now we will define process $Q$, which is $Q'$ modified such that $q_{0,i}^{1,i} = 0$ for all $i \in \{1, 2, \dots, n-2\}$. Then Corollary \ref{cor:domination} applies when the process starts from $(0,1)$, and gives $f((0,1);Q') \geq f_{\max}((0,1);Q) = f((0,1);Q)$, where the equality comes from the fact that, essentially, when we start from $(0,1)$, $Q$ is identical to $Q_{\max}$: observe that when $Q$ starts from $(0,1)$, it is equivalent to the Markov chain depicted in Figure \ref{UST}, where the states $(1,i-1)$ for all $i \in \{1, 2, \dots, n-1\}$ are removed, since they can never be reached. The latter inequality, combined with \eqref{eq: P,Q'} in \eqref{eq: f(P)} gives 
		\begin{align} \label{eq: f(P) > f(Q)}
			f(P) \geq \left(1-\frac{1}{n}\right) f((0,1);Q) .
		\end{align}
		
		What remains is to show the required lower bound on $f((0,1);Q)$. $Q$ starting from $(0,1)$ is a random walk on $\{(0,0),(0,1), \dots, (0,n-1), (1,n-1)\}$, or equivalently, by considering only the cardinality on the aforementioned states, the random walk is on $\{0,1,\dots, n\}$. The backward biases now are
		\begin{align*}
			&\gamma_k = \frac{1}{r} \cdot \frac{1}{n-k-1}  \quad \text{, for } 1 \leq k \leq n-2 \\
			\text{and } \quad &\gamma_{n-1} = \frac{1}{r} . 
		\end{align*}
		By using Fact \ref{remark2} for $i=1$ we get:
		\begin{align*}%\label{f_01}
			f((0,1);Q) &= \frac{1}{1 + \sum_{j=1}^{n-2} \prod_{k=1}^{j} \gamma_k + \frac{1}{r^{n-1}(n-2)(n-3)\cdots 1}} \nonumber \\
			&\geq \frac{1}{1 + \frac{1}{r(n-2)} + \frac{1}{r^{2}(n-2)(n-3)}\cdot(n-2)} \qquad \text{, for } r>(n-4)!^{-1/(n-2)} \nonumber \\
			&= \frac{1}{1 + \frac{1}{r(n-2)} + \frac{1}{r^{2}(n-3)}} ,
			%&\rightarrow 1 \quad \text{as} \quad n \rightarrow \infty.  \nonumber 
		\end{align*}
		and consequently, from \eqref{eq: f(P) > f(Q)} we get:
		\begin{align*}
			f(P) &\geq \frac{1-\frac{1}{n}}{1 + \frac{1}{r(n-2)} + \frac{1}{r^{2}(n-3)}} \\
			&> 1 - \frac{1}{n} - \frac{1}{r(n-2)} - \frac{1}{r^{2}(n-3)} \\
			&\rightarrow 1 \quad \text{as} \quad n \rightarrow \infty.
		\end{align*}
		This completes the proof of Theorem \ref{lemma1}.

	\end{proof}

	The latter result highlights an important difference between the single-graph and the two-graphs Moran process. Particularly, in the former, for $r=1$ any strongly connected digraph on $n$ nodes yields fixation probability $1/n$; this is by the observation that the proof of Lemma 1 of \cite{D14} holds for any strongly connected digraph. However, in the two-graphs process, we see that the structure of the graphs plays a central role in the magnitude of the fixation probability ($f(S_n , K_n) \to 1$ for any constant $r$). In fact, as it will become apparent from the results in this work, for any constant value of $r$, it could be that $f \to 0$ or $f \to 1$ depending on the pair of graphs under examination. Therefore, no constant value of $r$ can significantly affect $f$ in the two-graphs model.

	It is worth noting that, since the game we defined in the beginning of this section is a 1-sum game, we immediately can get upper (resp. lower) bounds on the payoff of player $R$, i.e. the extinction probability, by deducing lower (resp. upper) bounds on the payoff of player $M$, i.e. the fixation probability.

	The following lemma that connects the fixation probability of a process with given relative fitness, resident and mutant graphs, with the fixation probability of a ``mirror'' process where the roles between residents and mutants are exchanged.
	\begin{lem}\label{mirror}
		For any $r > 0$, and any pair $G_R$, $G_M$ of strongly connected directed graphs, $f(G_R, G_M ; r) \leq 1 - f(G_M, G_R ; 1/r)$.
	\end{lem}
	
	\begin{proof}
		We first prove the following:
		
		\begin{claim}\label{claim1}
			For any $r > 0$, and any pair $G_R$, $G_M$ of strongly connected directed graphs, $f^{S}(G_R, G_M ; r) = 1 - f^{V \setminus S}(G_M, G_R ; 1/r)$.
		\end{claim}
		\begin{proof}
			The probability of fixation for a mutant set $S$ and mutant graph $G_M$ is the same as the probability of extinction of the resident set $V \setminus S$, i.e. one minus the probability of the set $V \setminus S$ conquering the graph. Thus, if we exchange the labels of residents and mutants, the relative fitness of the new residents is 1 and the relative fitness of the new mutants is $1/r$, the new resident graph is $G_M$, the new mutant graph is $G_R$ and the new mutant set is $V \setminus S$.  
		\end{proof}

		We can now prove Lemma \ref{mirror} as follows:	
		By the above claim we have $f^{\{u\}}(G_R, G_M ; r) = 1 - f^{V \setminus \{u\}}(G_M, G_R ; 1/r)$ for every $u \in V$. 
		Since by Corollary \ref{cor:subset_dom} we have $f^{\{v\}}(G_M, G_R ; 1/r) \leq f^{V \setminus \{u\}}(G_M, G_R ; 1/r)$ for every $v \neq u$, we get that $f^{\{u\}}(G_R, G_M ; r) \leq 1 - f^{\{v\}}(G_M, G_R ; 1/r)$. 
		Let $\pi$ be a derangement of $V$, and let $\pi(u)$ denote the node $v \neq u$ that is in the derangement at the corresponding place of $u$. Then, starting from the definition of the fixation probability, we have
%		Averaging over all nodes in $V$ we get the required inequality. This completes the proof of Lemma \ref{mirror}. 
%		\begin{align}
%		f_{G_R, G_M}(r) &= \frac{1}{n}\sum_{u\in V} f_{G_R, G_M}^{\{u\}}(r)                \\
%		&= \frac{1}{n}\sum_{u\in V} \left(1-f_{G_M, G_R}^{\{\pi(u)\}}(1/r)\right) \quad \text{(by Claim \ref{claim1})}  \\
%		&\leq \frac{1}{n}\sum_{u\in V} \left(1-f_{G_M, G_R}^{\{u\}}(1/r)\right) \quad \text{(by Lemma \ref{mirror})}          \\
%		&= 1 - f_{G_M, G_R}(1/r).                                         
%		\end{align}
%
		\begin{align*}
		f(G_R, G_M ; r) &= \frac{1}{n}\sum_{u\in V} f^{\{u\}}(G_R, G_M ; r)                \\
		&\leq \frac{1}{n}\sum_{u\in V} \left(1-f^{\{\pi(u)\}}(G_M, G_R ; 1/r)\right)  \quad \text{(see above paragraph)} \\
		&= \frac{1}{n}\sum_{v\in V} \left(1-f^{\{v\}}(G_M, G_R ; 1/r)\right)          \\
		&= 1 - f(G_M, G_R ; 1/r).                                         
		\end{align*}
		
	\end{proof}
	Lemma \ref{mirror} provides easily an upper bound on the fixation probability of a given process when a lower bound on the fixation probability is known for its ``mirror'' process. For example, using Theorem \ref{lemma1} and Lemma \ref{mirror} we get an upper bound $\frac{1}{n} + \frac{r}{(n-2)} + \frac{r^{2}}{n-3}$ for $r < (n-4)!^{1/(n-2)}$ on the fixation probability of $K_n$ vs $S_n$; this immediately implies that the probability of fixation in this case tends to 0 as $n \to \infty$ for $r \in o(\sqrt{n})$. However, as we subsequently explain, a more precise lower bound is necessary to reveal the approximation restrictions of that particular process.

	\begin{thm}\label{lemma5}
		If $G_R = K_n$ and $G_M = S_n$ for $r>0$, then the fixation probability is upper bounded by $\frac{r^{n-1}}{(n-2)!}$.
	\end{thm}
	
	\begin{proof}
		The proof is similar to that of Theorem \ref{lemma1}. Namely, we will find an upper bound on the fixation probability of our process $P$ by giving an upper bound on the following process $Q$ that dominates (has at least the fixation probability of) it. 
		
		The Markov chains of $P$ and $Q$ are depicted in Figure \ref{K-S}.
		\begin{figure}[H]
			\begin{center}
			\includegraphics[scale=0.68]{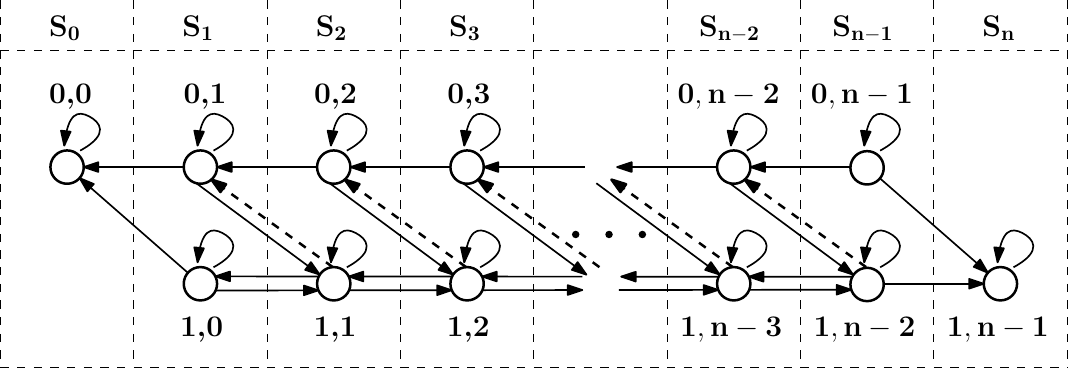}
			\caption{The Markov chains for process $P$ and $Q$. In $Q$, the dashed edges have probability $0$ (their former probability is added to the self loop).}\label{K-S}
			\end{center}
		\end{figure}
		Here we use the same notation as in the proof of Theorem \ref{lemma1} for the states of the induced Markov chains. 
%			
%		Here, again \eqref{eq: f(P)} holds, and for our required bound we will upper bound the fixation probabilities $f((0,1);P)$ and $f((1,0);P)$.
		%	
%		
	In $P$, the transition probabilities for $1 \leq i \leq n-1$ are:
	\begin{align*}
		p_{0,i}^{0,i-1} &= \frac{n-i}{ir+n-i} \cdot \frac{i}{n-1},  \quad p_{1,i-1}^{0,i-1} = \frac{n-i}{ir+n-i} \cdot \frac{1}{n-1},  \\
		p_{1,i-1}^{1,i-2} &= \frac{n-i}{ir+n-i} \cdot \frac{i-1}{n-1}, \quad \text{for } i \neq 1  , \\
		p_{0,i}^{1,i} &= \frac{ir}{ir+n-i}, \quad p_{1,i-1}^{1,i} = \frac{r}{ir+n-i} \cdot \frac{n-i}{n-1}.
	\end{align*}

	For the rest of the proof we will need to refer to the proof of Theorem \ref{lemma1} itself, not just the statement. In particular, first observe that starting process $P$ from any state in $S_1$, we cannot reach $(0,n-1)$, therefore, we can remove it from the Markov chain. Now let us define the modified process $Q'$ where the only difference from $P$ is that $q_{1,n-2}^{0,n-2} = 0$ (and $p_{1,n-2}^{0,n-2}$ is added to $q_{1,n-2}^{1,n-2}$). According to the proof of Theorem \ref{lemma1}, it holds that $f((0,1);P) \leq f((0,1);Q')$ and $f((1,0);P) \leq f((1,0);Q')$, thus the fixation probability $f(P)$ is upper bounded by $f(Q')$. Now observe that starting from any state $S_1$, $(0,n-2)$ is not reachable, and similarly to the case of $P$, we can remove it from the chain, leaving $(1,n-3)$ as the unique state of $S_{n-2}$.
	
	Again, we define another intermediate process $Q''$ in which $q_{1,n-3}^{0,n-3} = 0$, and similarly to $Q'$, according to the proof of Theorem \ref{lemma1} we have $f(Q') \leq f(Q'')$. By sequentially modifying the process such that eventually $q_{1,n-i}^{0,n-i} = 0$ for every $i \in \{2, 3, \dots, n-1 \}$, we end up with process $Q^{(n-2)}$, where $f(P) \leq f(Q^{(n-2)})$. Finaly, observe that for the cumulative backward biases holds that $\gamma^{c}((0,1);Q^{(n-2)}) = \gamma^{c}((1,0);Q^{(n-2)}) = 1/r$, therefore we can reduce the set $S_1$ of $Q^{(n-2)}$ to a singleton without changing the fixation probability. Without loss of generality we have $S_1 = \{(1,0)\}$ and let us call the resulting process $Q$. Finally, for ease of presentation, we identify the states of $Q$ with their cardinality, i.e. $\{0,1, \dots, n\}$, and $Q$ is clearly a random walk on these states with backward biases:
		\begin{align*}
		&\gamma_1 = \frac{1}{r} , \quad  \text{and,} \\
		&\gamma_k = \frac{1}{r} \cdot (k-1) , \quad \text{ for } 2 \leq k \leq n-1 . 
		\end{align*}
		The fixation probability of $Q$ can be found by using Fact \ref{remark2} for $i=1$:
		\begin{align*}%\label{up}
		f(Q) &= \frac{1}{1 + \frac{1}{r}1 + \frac{1}{r^2}1 + \frac{1}{r^3}2 + \frac{1}{r^4}3! + \cdots + \frac{1}{r^{n-1}}(n-2)!}  \nonumber   \\
		& \leq \frac{r^{n-1}}{(n-2)!}  \\
		& \in o\left(\frac{1}{a^n}\right) , \quad \text{ for any constant } a>1. \nonumber  
		\end{align*}
		This completes the proof of Theorem \ref{lemma5}. 
	\end{proof}

	This bound shows that not only there exists a graph that suppresses selection against $S_n$ (which is an amplifier in the single-graph setting), but it also does that with great success. In fact, the last row of the above expression shows that for any mutant with constant $r$ arbitrarily large, its fixation probability is less than exponentially small.

	In view of the above, the following result implies that the fixation probability in our model cannot be approximated via a method similar to \cite{D14}, because the lowest fixation probability among all pairs of graphs is less than exponentially small in the input size of the graphs.
	\begin{thm}[Bounds on the two-graphs Moran process]\label{th_bound}
		There is a pair of graphs $G_R, G_M$ such that the fixation probability $f(G_R, G_M)$ is $o\left(\frac{1}{a^n}\right)$, for constant $a>1$, when the relative fitness $r$ is constant. Furthermore, there is a pair of families of graphs $G_R', G_M'$ such that the fixation probability $f(G_R', G_M')$ is at least $1-O\left(\frac{1}{n}\right)$, for constant $r>0$.
	\end{thm}
	
	\begin{proof}
		See Theorem \ref{lemma1} and proof of Theorem \ref{lemma5}.  
	\end{proof}

	\subsection{Arbitrary Undirected Graphs vs Clique}
	
	The following result is a lower bound on the fixation probability. We remind that, according to Section \ref{definitions}, a resident (resp. mutant) edge between nodes $x,y \in V$ of an undirected graph represents two distinct positive weights, $w_{xy}^R, w_{yx}^R$ (resp. $w_{xy}^M, w_{yx}^M$), and should not be confused with typical graph-theoretic settings where an undirected edge is associated with a single weight. We remind that under the current definition of weights, if nodes $x$ and $y$ of type $T_x , T_y \in \{ R, M \}$ respectively have different out-degree, then $w_{xy}^{T_x} \neq w_{yx}^{T_y}$.
	
	\begin{thm}\label{LB}
		When the resident graph is an undirected graph for which $w_{xy}^R/w_{yx}^R \leq c$ for every $(xy) \in E_R$ and the mutant graph is $K_n$, the fixation probability is lower bounded by $\left[\frac{1-\left(\frac{c}{r}\right)^{g(n)}}{1-\frac{c}{r}}\left(1 + o(1) \right) + \frac{\left(\frac{2c}{r}\right)^{g(n)}-\left(\frac{2c}{r}\right)^{n}}{1-\frac{2c}{r}}\right]^{-1}$, for $r>0$ and any function $g(n) \in \omega(1) \cap o(n)$. In particular, for $r>2c$ the lower bound tends to $1 - \frac{c}{r}$ as $n \to \infty$.
	\end{thm}
	
	\begin{proof}
		Notice that, given the number of mutants at a time-step is $i:=|S|$, the probability that some resident becomes mutant is $p_{i}^{i+1} = \frac{ir}{ir+n-i} \cdot \frac{n-i}{n-1}$, and the probability that some mutant becomes resident $p_{i}^{i-1}$ is upper bounded by $\frac{\min\{i,n-i\}}{ir+n-i} \max\limits_{(xy) \in E_R} \frac{w_{xy}^R}{w_{yx}^R}$. The following reasoning explains this upper bound. 
		
		First, we remind the definition of $p_{i}^{i-1}$ from equation \eqref{fixation1}, i.e. the probability of the middle term of the right-hand side of the equation: $p_{i}^{i-1} = \frac{1}{ir + n-i} \sum\limits_{x \notin S, y \in S} w_{xy}^{R}$. The sum in the latter formula is the sum of weights that the mutant nodes $y$ receive from resident nodes $x$. We will upper bound this sum of weights that go towards all mutants from residents. Each mutant node $y$ receives weight from at most $out-degree^{R}(y) = 1/w_{yx}^R$ residents (by definition of $w_{yx}^R$ for any $x$). The weight that comes from resident $x$ equals $w_{xy}^R$. Therefore, in total, the maximum weight that $y$ receives is $c:= \max\limits_{(xy) \in E_R}\frac{w_{xy}^R}{w_{yx}^R}$, where $x$ is the maximum weight (or equivalently, lowest degree) resident node that is neighbour of $y$ in $G_R$. 
		It remains to sum over all mutants $y$. For this we use a symmetry trick. Observe that in the way we have bounded the total incoming weight of a mutant from its resident neighbours, only the weights in the resident graph $G_R$ matter. There is a flow of weights from residents to mutants, but since this flow considers only the resident edges for both residents and mutants, by exchanging the roles of mutants and residents for the nodes, the flow should remain the exact same. To see this, note that between mutant nodes $z,w$, mutant-to-mutant edges (in $E_R$) will become resident-to-resident edges and vice versa, and obviously, the sum of $w_{zw}^R + w_{wz}^R$ will remain the same. This will happen for all pairs of nodes of the same type, and since the sum of all weights of all $n$ nodes equals $n$, the sum of weights $w_{xy}^R + w_{yx}^R$ of nodes $x$ and $y$ of opposite type will remain the same. Therefore, we can sum over all mutants $y$, or we can exchange the roles of mutants and residents and then sum over all mutants, and both these upper bounds are fine. Therefore, the minimum upper bound can be used, which is $\min\{i,n-i\}$. 
		
		Observe now that Lemma \ref{lem: domination} applies and our given process $P$ of an arbitrary undirected graph vs Clique stochastically dominates a birth-death process $Q$ that is described by the following Markov chain: A state $i$, where $i \in \{0,1,2,...,n\}$ is the number of mutants on the vertex set and the only absorbing states are $0$ and $n$. In particular, $f(S;P) \geq f_{\max}(S;P) = f(S;Q)$ for any given mutant set $S \subseteq V$. The aforementioned transition probabilities of our Markov chain allow us to bound the backward bias defined in Fact \ref{remark2}:
		\[
		\gamma_k \leq \begin{cases}
		\frac{c}{r} \cdot \frac{n-1}{n-k}, & \text{for }  k \in \{1,2,..,\left\floor{\frac{n}{2}\right}\}\\
		\frac{c}{r} \cdot \frac{n-1}{k}, & \text{for }  k \in \{\left\floor{\frac{n}{2}\right}+1,..,n-1\}
		\end{cases}
		\]
		Now we can calculate a lower bound on the fixation probability of $Q$ using Fact \ref{remark2} for $i=1$, the fact that $\frac{n-1}{n-a} = 1 + \frac{a-1}{n-a}$, for all $a \neq n $, and by considering an arbitrary function $g(n) \in \omega(1) \cap o(n)$ we have:
		\begin{align*}%\label{f_1_lo}
		f(Q) &= \frac{1}{\left[\sum_{j=0}^{g(n) - 1}\left(\frac{c}{r}\right)^j\right]\left( 1 + o(1) \right) + \sum_{j=g(n)}^{n-1} \prod_{k=1}^{j} \gamma_k } \nonumber   \\
		& \geq \frac{1}{\frac{1-\left(\frac{c}{r}\right)^{g(n)}}{1-\frac{c}{r}}\left( 1 + o(1) \right) + \left(\frac{2c}{r}\right)^{g(n)}\sum_{j=0}^{n - g(n) - 1}\left(\frac{2c}{r}\right)^j} \quad \left(\text{since } \gamma_k \leq \frac{2c}{r}\right) \nonumber  \\
		&= \frac{1}{\frac{1-\left(\frac{c}{r}\right)^{g(n)}}{1-\frac{c}{r}} \left(1 + o(1) \right) + \left(\frac{2c}{r}\right)^{g(n)}\frac{1-\left(\frac{2c}{r}\right)^{n-g(n)}}{1-\frac{2c}{r}}}. 
		\end{align*}  
	\end{proof}

	From the theorem above it follows that if $G_R$ is an undirected regular graph then the fixation probability of $G_R$ vs $K_n$ is lower bounded by $1-1/r$ for $r>2$ and $n \to \infty$, which equals $f_\text{Moran}$ (defined in Section \ref{definitions}). 
	
	Also, by Lemma \ref{mirror} and the above theorem, when $G_R = K_n$, $G_M$ is an undirected graph with $w_{xy}^M/w_{yx}^M \leq c$ for every $(xy) \in E_M$, and relative fitness $r<\frac{1}{2c}$, then the upper bound of the fixation probability tends to $cr$ as $n \to \infty$. In particular, if $G_M$ is an undirected regular graph, then
	the fixation probability is upper bounded by $r$ for $r < \frac{1}{2}$ and $n \to \infty$.

	\subsection{Circulant Graphs vs Clique} \label{classesvsCL}
	
	In this subsection we give bounds for the fixation probability of $CId$ vs $K_n$. We first prove the following result that gives an upper bound on the fixation probability when $G_R$ is the $CId$ graph as described in Section \ref{definitions} and $G_M$ is the complete graph on $n$ nodes.

	%\begin{thm}\label{lemma7}
	%	 Let $d$ be the parameter of the graph $CId$, and let it be a function of $n$. When mutants have the $CL$ graph, if residents have a $CId$ graph and $d \in \Theta(1)$, then the payoff of pl-M (fixation probability) is upper bounded by $\left[e^{\frac{1}{r}} - \frac{1}{r^{n}}\frac{1}{n!} \frac{1}{1-\frac{1}{r}}\right]^{-1}$ for $r>1$ and $\left[e^{\frac{1}{r}} - \frac{1}{r^{n}}\frac{1}{n!} - o(1)\right]^{-1}$ for $r \leq 1$. In particular, for constant $r>0$ the upper bound tends to $e^{-\frac{1}{r}}$. If $d \in \omega(1)$, then the upper bound is $\left(1-\frac{1}{r}\right)\left[1-\frac{1}{r^{g(n)}} - o(1)\right]^{-1}$, for $r>0$, where $g(n)$ is a function of $n$ such that $g(n) \in \omega (1)$ and $g(n) \in o(d)$. The bound improves as $g(n)$ is picked closer to $\Theta (d)$ and, in particular, for $r>1$ it tends to $1-\frac{1}{r}$.
	%\end{thm}
	\begin{thm}\label{lemma7}
		Let $d$ be the degree of the graph $CId$, and let it be a function of $n$. When $G_R = CId$ and $G_M = K_n$: 
		\begin{itemize}
			\item If $d = d(n) \in \Theta(1)$, then the fixation probability is upper bounded by $e^{-1/r}$ for $r \geq \frac{1}{\sqrt{n}}$.
			\item If $d = d(n) \in \omega(1)$, then the fixation probability is upper bounded by $\left( 1 - \frac{1 - o(1)}{r}\right) \left( 1 - \left( \frac{1 - o(1)}{r} \right)^{g(n)+2} \right)^{-1}$, for $r>0$ and any function $g(n) \in \omega (1) \cap o(d)$ such that $g(n) \leq d/2$ for the given $n$. In particular, for $r>1$ the bound tends to $1-\frac{1}{r}$ as $n \to \infty$.
		\end{itemize}
	\end{thm}
	
	\begin{proof}
		We will bound from above the fixation probability of our process $P$, by giving an upper bound on the fixation probability of a process $Q$ that dominates $P$ ($f(P) \leq f(Q)$). In $Q$, given that the number of mutants is $i$ at some step, the positioning of the mutants is such that minimizes the transition probability to a state with $i-1$ mutants. Then, by Lemma \ref{lem: domination} the required domination follows.
		
		Here is $Q$: Have the $CId$ graph for the residents, as defined in Section \ref{definitions}, and the clique graph for the mutants. Let us label the nodes of $CId$ as shown in Figure \ref{CIdvsCL}. The process starts in a way identical to that of $P$, i.e. with a single mutant on a node (without loss of generality we give it label $1$) chosen uniformly at random from the vertex set. Throughout the process, if a resident is selected to reproduce onto a resident, or a mutant is selected to reproduce onto a mutant, it reproduces according to the exact same rules of $P$. However, if a mutant is selected to reproduce onto a resident, instead of reproducing on it, it reproduces on an arbitrary resident that is adjacent to the maximum number of mutants possible in the current state. If a resident is selected to reproduce onto a mutant when the number of mutants is $i \in \{1,2,...,n-1\}$, then the last among the $i$ mutants that was inserted becomes resident, thus preserving the minimality of the probability of the residents to hit the mutants (see Figure \ref{CIdvsCL}). Without loss of generality, $Q$ is the process where at any step with mutant set $S = \{ 1, 2, \dots, |S| \}$: the first mutant is at node 1, the most recently inserted mutant occupies node $|S|$, and the new mutant occupies node $|S|+1$. 
		
		\begin{figure}[t]
			\begin{center}
				\includegraphics[scale=0.45]{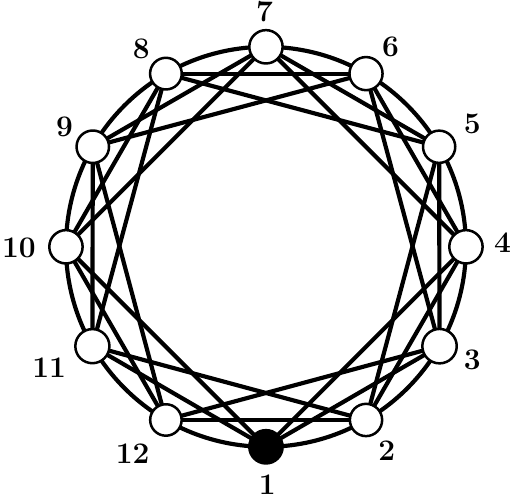}
				\includegraphics[scale=0.45]{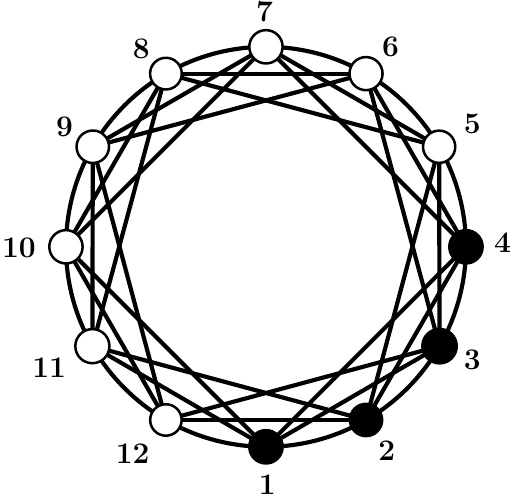}
				\includegraphics[scale=0.45]{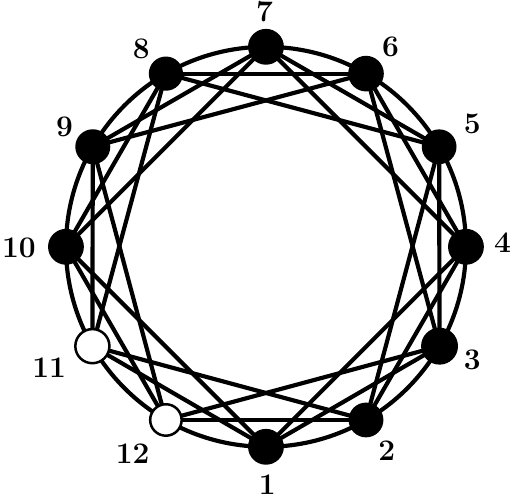}
				\caption{Three instances of process $Q$. Only the resident graph $G_R = CId$ is shown. White nodes are residents and black nodes are mutants. Here $d=6$ and $n=12$.} \label{CIdvsCL}
			\end{center}
		\end{figure}
		
		The weight of every resident edge is $1/d$, therefore, in a state with $i$ mutants, in order to minimize the sum of resident edges incident to a resident and a mutant it suffices to show that the number of such edges, is minimized (over all allocations of fixed $i$ mutants) by $Q$: this is the case since by inserting one by one the mutants, the closest the inserted mutant is to the lastly inserted mutant, the most mutant-to-mutant edges it creates and the less resident-to-resident edges it destroys, and therefore, the more mutant-to-resident edges it yields. Finally, by definition of process $Q$, at any time $t \geq 0$ the mutant-to-resident edges are at most as many as they are in $P$. In fact, $Q$ is the modified process $P$, in which, for each set of states $S_k$ we have kept only the state $S_{\min}^{(k)} \in S_k$ where mutants are placed on nodes $1,2,\dots, k$. This state minimizes the cumulative backward bias, since $S_{\min}^{(k)} = \arg \min_{X \in S_{k}} \sum_{j \in S_{k-1}} p_{X}^{j}$ and for every $X \in S_k$ we have $\sum_{j \in S_{k+1}} p_{S_{\min}^{(k)}}^{j} = \sum_{j \in S_{k+1}} p_{X}^{j}$. Then, $f(S;P) \leq f_{\min}(S;P) = f(S;Q)$ for any initial state $S \subseteq V$, where the inequality comes from Lemma \ref{lem: domination}.
%		In other words, if we consider the mutant set $S$ and the resident set $V \setminus S$, in every step of the process the number of resident edges on the cut $(S, V \setminus S)$ of $G_R$  is minimum. Since the probability of making a mutant resident is minimized, this process is the worst the residents could deal with.
		
		What remains to upper bound the fixation probability of $Q$. Due to the symmetry that our process $Q$ brings on the population instances, the corresponding Markov chain has $n+1$ states, as every state with the same number of mutants can be reduced to a single one. A state $i$, where $i \in \{0,1,2,...,n\}$ is the number of mutants and the only absorbing states are $0$ and $n$. After careful calculations we get that, for a state $i$, where $i \in \{1,2,...,n-1\}$, the probability of going to state $i-1$ in the next step is:
		\begin{align}\label{eq:backw_prob}
		p_{i}^{i-1} =\begin{cases}
		\frac{1}{ir+n-i} \cdot i \left( 1 - \frac{i-1}{d} \right), & \text{for }  i \in \{1,2,..,\frac{d}{2}+1\}\\
		\frac{1}{ir+n-i} \cdot \frac{1}{2} \cdot \left( \frac{d}{2}+1 \right), & \text{for }  i \in \{\frac{d}{2}+2,..,n-\frac{d}{2}\}\\
		\frac{1}{ir+n-i} \cdot (n-i) \cdot (1-\frac{n-i-1}{d}), & \text{for }  i \in \{ n-\frac{d}{2}+1,..,n-1\},
		\end{cases}
		\end{align}
		where the last formula is due to the fact that the residents send out $(n-1)d$ edges, of which $(n-i)(n-i-1)$ are to other residents, therefore in this case $p_i^{i-1} = \frac{1}{ir+n-i}\left(\frac{(n-i)d}{d}-\frac{(n-i)(n-i-1)}{d}\right)$.
		The probability of going to state $i+1$ in the next step is:
		\begin{align*}
		p_{i}^{i+1} = \frac{ir}{ir+n-i} \cdot \frac{n-i}{n-1} .
		\end{align*}
		and the probability of remaining in state $i$ in the next step is: $p_{i}^{i} = 1-p_{i}^{i-1}-p_{i}^{i+1}$. 
		
		If \pmb{$d \leq g(n)$} for some arbitrary fixed function $g(n) \in \Theta(1)$, we bound from above 
		the fixation probability $f(Q)$ given that we start from state $1$. To this end we use Fact \ref{remark2} for $i=1$, and the general lower bound on $\gamma_k$ yielded by the fact that the right-hand sides of \eqref{eq:backw_prob} are at least $(ir + n - i)^{-1}$:
		\begin{align*}
		\gamma_k \geq \frac{1}{r} \cdot \frac{n-1}{k(n-k)}  \quad \text{, for } 1 \leq k \leq n-1 .
		\end{align*}
		We have
		\begin{align*}
		f(Q) &\leq \frac{1}{1 + \frac{1}{r}\frac{n-1}{n-1} + \frac{1}{r^2}\frac{(n-1)^2}{2(n-1)(n-2)} + \frac{1}{r^3}\frac{(n-1)^3}{3!(n-1)(n-2)(n-3)} + \cdots + \frac{1}{r^{(n-1)}}\frac{(n-1)^{(n-1)}}{[(n-1)!]^2}}  \nonumber   \\
		& \leq \frac{1}{1 + \frac{1}{r} + \frac{1}{r^2}\frac{1}{2} + \frac{1}{r^3}\frac{1}{3!} + \cdots + \frac{1}{r^{(n-1)}}\frac{1}{(n-1)!}}  \nonumber   \\
		& = \frac{1}{e^{\frac{1}{r}} - \frac{1}{n!} e^{\xi} \frac{1}{r^{n}}}, \quad \text{ for some $\xi \in (0, \frac{1}{r})$,}  \nonumber  \\
		&= e^{-\frac{1}{r}} + o(1) , \quad \text{as $n \rightarrow \infty$}, \quad \text{ for $r \geq \frac{1}{\sqrt{n}}$}  \quad \text{(by Stirling's formula),}
		\end{align*}
		where the second to last equation comes from the expression of Maclaurin series of $e^{1/r}$.
		
		For the case \pmb{$d = d(n) \in \omega(1)$} we will employ an arbitrary fixed function $g(n) \in \omega(1) \cap o(d)$, such that additionally $g(n) \leq d/2$ for the given $n$. Again we give an upper bound on $f(Q)$. Now, it seems that the probability to reach state $d/2 + 1$ given that we start from state $1$ is close enough to the fixation probability, hence we bound $f(Q)$ by stopping at the former state and considering that we have fixation. Since we are considering this case, we use the more precise bound on $\gamma_k$ yielded by the top formula of \eqref{eq:backw_prob}:
		\begin{align*}
			\gamma_k \geq \frac{1}{r} \left( 1 - \frac{k-1}{d}\right) .
		\end{align*}
		
		We have
		\begin{align*}%\label{f_1_up3}
		f(Q) & \leq \frac{1}{1+\sum_{j=1}^{d/2+1}\prod_{k=1}^{j}\gamma_k} \\
		& \leq \frac{1}{1 + \sum_{j=1}^{d/2+1}\prod_{k=1}^{j} \frac{1}{r}\left( 1 - \frac{k-1}{d} \right)} \nonumber   \\
		&\leq \frac{1}{1+\sum_{j=1}^{g(n)+1}\prod_{k=1}^{j} \frac{1}{r}\left( 1 - \frac{k-1}{d} \right)} \quad \text{ (by definition of $g(n)$) }  \nonumber  \\
		&\leq \frac{1}{1+\sum_{j=1}^{g(n)+1} \frac{1}{r^j}\left( 1 - \frac{g(n)}{d} \right)^j}  \nonumber  \\
		&= \frac{1 - \frac{1 - \frac{g(n)}{d}}{r}}{1 - \left( \frac{1 - \frac{g(n)}{d}}{r} \right)^{g(n)+2}}   \nonumber  \\
		&= \frac{1 - \frac{1 - o(1)}{r}}{1 - \left( \frac{1 - o(1)}{r} \right)^{g(n)+2}} , \quad \text{as $n \rightarrow \infty$}, \quad \text{for $r>0$}. \nonumber 
		\end{align*}
		Recall now that for the initial Moran process $P$ it holds that $f(P) \leq f(Q)$. This completes the proof of Theorem \ref{lemma7}.  
	\end{proof}

	Combining the latter result with the lower bound on the fixation probability from Theorem \ref{LB}, we get the following corollary.

	\begin{corollary}
		Let mutants have the $K_n$ graph and residents have the $CId$ graph with degree $d$ which is a function of $n$. Then, for constant $r>2$ and $n \to \infty$, the fixation probability tends to 
		\begin{itemize}
			\item a value between $1 - \frac{1}{r}$ and $e^{-\frac{1}{r}}$ when $d=d(n) \in \Theta(1)$,
			\item $1 - \frac{1}{r}$ when $d=d(n) \in \omega(1)$.
		\end{itemize}
	\end{corollary}

	\section{An Approximation Algorithm}\label{sec:fpras}
	
	Here we present a fully polynomial randomized approximation scheme (FPRAS)\footnote{An FPRAS for a function $f$ that maps problem instances to numbers is a randomized algorithm with input $I$ and parameter $\epsilon > 0$, which has running time polynomial in $|I|$ and $\epsilon ^{-1}$ and outputs a random variable $g$, such that $\Pr\{(1-\epsilon)f(I) \leq g(I) \leq (1+\epsilon)f(I)\} \geq \frac{3}{4}$ \cite{KL83}.} for the problem \textsc{UndirectedVsClique} of computing the fixation probability in the Moran process when the residents have an undirected graph and the mutants have the clique graph with $r>2c \left(1+\frac{1}{n-2}\right)$, where $c$ is the maximum ratio of the degrees of adjacent nodes in the resident graph. The following result is essential for the design of a FPRAS; it gives an upper bound (which depends on $c, r$ and is polynomial in $n$) on the expected absorption time of the Moran process in this case. 
	\begin{thm}\label{tau_upbound}
		Let $G_{R} = (V,E_R)$ be an undirected graph, for which $w_{xy}^R/w_{yx}^R \leq c$ for every $(xy) \in E_R$ and let $G_M = (V,E_M)$ be the clique graph, where $|V| = n \geq 3$. For $r \geq 2c\left( 1 + \frac{1}{n-2} \right)$ and any $S \subseteq V$, the absorption time $\tau$ of the ``$G_R$ vs $G_M$" Moran process $(X(t))_{t \geq 0}$ satisfies:
		\begin{align*}
		\mathbb{E}[\tau | X(0) = S] \leq \frac{r}{r-c} n^2. %\leq \frac{r}{r-c} n^2.
		\end{align*}
	\end{thm}
	
	\begin{proof}
		For any possible mutant set $S \subseteq V$, consider the number of mutants $|S|$. We first prove the following intermediate result which states that the number of mutants strictly increases in expectation when $r \geq 2c\left( 1 + \frac{1}{n-2} \right)$, where $c:= \max\limits_{(xy) \in E_R}\left(w_{xy}^R/w_{yx}^R\right)$, i.e. the maximum ratio of the degrees of adjacent nodes in the resident graph.
		
		\begin{lem}\label{exp_bound}
			Let $(X(t))_{t \geq 0}$ be a Moran process in which the resident graph is a general undirected graph and the mutant graph is the clique. If $r \geq 2c\left( 1 + \frac{1}{n-2} \right)$, where $c:= \max\limits_{(xy) \in E_R}\left(w_{xy}^R/w_{yx}^R\right)$, then
			\begin{align*}
			\mathbb{E}\left[|X(t+1)| - |X(t)| \text{ } | \text{ } X(t) = S \right] \geq \left(1 - \frac{c}{r}\right)\frac{1}{n},
			\end{align*}
		where $\emptyset \subset S \subset V$ is the mutant set.
		\end{lem}

		\begin{proof}
			By $F(S)$ we denote the total fitness $|S|r + |V| - |S|$ of a state with mutant set $S$. For $r \geq 2c\left( 1 + \frac{1}{n-2} \right)$ and $\emptyset \subset S \subset V$ we have:
			\begin{align*}
			& \mathbb{E}\left[|X(t+1)| - |X(t)| \text{ } | \text{ } X(t) = S \right] = \\
			& = \frac{1}{F(S)} \left[ \sum_{\substack{xy \in E_M \\ x \in S, y \in V \setminus S}} r \cdot w_{xy}^{M} \cdot \left(|S+y| - |S|\right) + \sum_{\substack{yx \in E_R \\ x \in S, y \in V \setminus S}} w_{yx}^{R} \cdot \left(|S-x| - |S|\right) \right] \\
			& = \frac{1}{F(S)} \left[ \sum_{\substack{xy \in E_M \\ x \in S, y \in V \setminus S}} r \cdot w_{xy}^{M}  - \sum_{\substack{yx \in E_R \\ x \in S, y \in V \setminus S}} w_{yx}^{R} \right] \\
			& \geq \frac{1}{F(S)} \left[ \frac{r}{n-1} \cdot |S|(n-|S|)  - c \cdot \min\{|S|,n-|S|\} \right],
			\end{align*}
			where we used the fact that $w_{xy}^{M} = \frac{1}{n-1}$ for every $(xy) \in E_M$, and \\ $\sum\limits_{\substack{yx \in E_R \\ x \in S, y \in V \setminus S}} w_{yx}^{R} \leq c \cdot \min\{|S|,n-|S|\}$ as in the proof of Theorem \ref{LB}. 
			
			The maximum value of $F(S)$ is $rn$, and the minimum of the function in the last brackets is $r-c$ for $r \geq 2c\left( 1 + \frac{1}{n-2} \right)$. The latter is shown by the following claim. 
			
			\begin{claim}
				Let $n \geq 3$, $c \geq 1$, and $r \geq 2c\left( 1 + \frac{1}{n-2} \right)$. Then the function $f(m) = \frac{r}{n-1} \cdot m(n-m)  - c \cdot \min\{m,n-m\}$, for $m \in \{1,2, \dots, n-1\}$ has global minimum $r-c$ attained at $m=1$.
			\end{claim}
			
			\begin{proof}
				We will transfer $f$ to the real domain $[1, n-1]$ in order to derive its behaviour in terms of global minima, and then get the solution as required in the discrete domain. $f(m)$ is clearly symmetric about $m = n/2$, and therefore it suffices to find its global minimum when only the domain $m \in [1,n/2]$ is considered. We have, $f(m)= \frac{r}{n-1} \cdot m(n-m)  - c \cdot m$. Its second derivative is $-\frac{2r}{n-1} < 0$ therefore its global minimum is attained at $m=1$ or $m=n/2$. It holds that $r \geq 2c\left( 1 + \frac{1}{n-2} \right)$ if and only if $f(1) \leq f(n/2)$, since $f(1)=r-c$ and $f(n/2)= r \frac{n^2}{4(n-1)} - c \frac{n}{2}$.  
%				\paragraph{Case 2} $\mathbf{m \geq n/2}$. In this case, $f(m)= \frac{r}{n-1} \cdot m(n-m)  - c \cdot (n-m)$. Its second derivative is again $-\frac{2r}{n-1} < 0$ therefore its global minimum is attained at $m=n/2$ or $m=n-1$. It holds that $r \geq 2c\left( 1 + \frac{1}{n-2} \right)$ if and only if $f(n-1) \leq f(n/2)$, since $f(n-1)=r-c$ and $f(n/2)= r \frac{n^2}{4(n-1)} - c \frac{n}{2}$. 
			\end{proof}

			This completes the proof of Lemma \ref{exp_bound}. 
		\end{proof}
		
		The expected absorption time can be bounded using martingale techniques. In particular, we employ the following theorem from \cite{D14}, which was used to bound the absorption time of the process with a single undirected graph.
		
		\begin{thm}[Theorem 6, \cite{D14}] \label{tau_bound}
			Let $(Y(t))_{t \geq 0}$ be a Markov chain with state space $\Omega$, where $Y(0)$ is chosen from some set $I \subseteq \Omega$. If there are constants $k_1, k_2 > 0$ and a non-negative function $\psi : \Omega \to \mathbb{R}$ such that:
			\begin{itemize}
				\item $\psi(S) = 0$ for some $S \in \Omega$,
				\item $\psi(S) \leq k_1$ for all $S \in I$, and
				\item $\mathbb{E}[\psi(Y(t))-\psi(Y(t+1)) | Y(t) = S] \geq k_2$ for all $t \geq 0$ and all $S$ with $\psi(S) > 0$,
			\end{itemize}
			then $\mathbb{E}[\tau] \leq k_{1}/k_{2}$, where $\tau = \min\{t : \psi(Y(t)) = 0\}$.
		\end{thm}

		We can now prove Theorem \ref{tau_upbound} in a similar way as Theorem 7 of \cite{D14}: Let $(Y(t))_{t \geq 0}$ be the process that behaves identically to the Moran process $(X(t))_{t \geq 0}$ except that, if $Y(t) = \emptyset$ then $Y(t+1) = \{x\}$, where $x$ is a node chosen uniformly at random from $V$. Setting $\tau' = \min\{t : Y(t) = V \}$ as the absorption time of this process, we have $\mathbb{E}[\tau | X(0) = S] \leq \mathbb{E}[\tau' | Y(0) = S]$. By defining $\psi(S) = n - |S|$, $k_1 = n, k_2 = (1-c/r)/n$ we satisfy the first condition of the theorem for $S = V$, the second condition because $n - |S| \leq n$ for all $S \subseteq V$, and the third condition by Lemma \ref{exp_bound} since for all $t \geq 0$ and $S \subset V$,
		\begin{align*}
			\mathbb{E}[\psi(Y(t))-\psi(Y(t+1)) | Y(t) = S] &= \mathbb{E}[n - |Y(t)|- (n - |Y(t+1)|) | Y(t) = S]   \\
			&= \mathbb{E}[|Y(t+1)| - |Y(t)|| Y(t) = S]   \\
			&\geq \mathbb{E}[|X(t+1)| - |X(t)|| X(t) = S]  \\
			&\geq \left(1 - \frac{c}{r}\right)\frac{1}{n} ,
		\end{align*} 
		where the first inequality is strict only for $S = \emptyset$, by definition of process $(Y(t))_{t \geq 0}$.
		
		Finally, combining the above with Theorem \ref{tau_bound} we get $\mathbb{E}[\tau | X(0) = S] \leq \mathbb{E}[\tau' | Y(0) = S] \leq k_{1}/k_{2} = \frac{r}{r-c} n^2$.   
	\end{proof}

	For our algorithm to run in time polynomial in the length of the input, $r$ must be encoded in unary. 
	
	\begin{thm} \label{thm11}
		There is an FPRAS for \textsc{UndirectedVsClique}, for $r>2c \left(1+\frac{1}{n-2}\right)$.
	\end{thm}
	
	\begin{proof}
		We present the following algorithm. First, we find the constant $c$ by checking every edge of the resident graph and exhaustively finding the maximum ratio of adjacent nodes' degrees in $O(n^3)$ time. If and only if $r$ is greater than $2c \left(1+\frac{1}{n-2}\right)$, we simulate the Moran process where residents have the given undirected graph and mutants have the clique graph. We compute the proportion of simulations that reached fixation for $N=\left\ceil{2 \epsilon^{-2} \ln 16\right}$ simulation runs with maximum number $T=\left\ceil{e \frac{r}{r-c} n^2 \ln (8N) \right}$ of steps each. In case the simulation has not reached absorption by the $T$-th step, it stops and returns an error value.
		
		Also, each transition of the Moran process can be simulated in $O(1)$ time. This is possible if we keep track of the resident and mutant nodes in an array, thus choose the reproducing node in constant time. Further, we can pick the offspring node in constant time by running a breadth-first search for each graph before the simulations start, storing the neighbours of each node for the possible node types (resident and mutant) in arrays. Hence the total running time is $O(n^3 + NT)$, which is polynomial in $n$ and $\epsilon^{-1}$ as required by the FPRAS definition.
		
		Now, we only have to show that the output of our algorithm computes the fixation probability to within a factor of $1\pm \epsilon$ with probability at least $3/4$. Essentially, the proof is the same as in \cite{D14} with modifications needed for our setting. For $t \in \{1,2,...,N\}$, let $Y(t)$ be the indicator variable, where $Y(t)=1$ if the $t$-th simulation of the Moran process reaches fixation and $Y(t)=0$ otherwise. We first calculate the bounds on the probability of producing an output of error $\epsilon$ in the event where all simulation runs reach absorption within $T$ steps. The output of our algorithm is then $g = \frac{1}{N}\sum_{t=1}^{N}Y(t)$ while the required function is the fixation probability $f$. Since $f$ is the mean value of $g$, by using Hoeffding's inequality \cite{H63} we get:
		\begin{align}\label{eq:err_bound}
		Pr\{|g-f|>\epsilon f \} \leq 2 e^{-2\epsilon^2 f^2 N} \leq 2 e^{-4 f^{2} \ln 16} < \frac{1}{8}, 
		\end{align}
		where the latter inequality is because $f \geq 1-c/r > 1/2$ due to Theorem \ref{LB}.

		Next, we will bound the probability of not reaching absorption within $T$ steps. For this, we will employ the following lemma from \cite{CIN17}, slightly modified so that instead of having a multiplicative factor of $2$ we have $e$. The term \textit{effective step} stands for a step of the process in which there is a node whose type changes.
		\begin{lem}[Lemma 4, \cite{CIN17}]
			Consider an upper bound $\ell$, for each mutant set $\emptyset \subset S \subset V$, on the expected number of (effective) steps to fixation. Then for any starting type function the probability of fixation requires more than $e \cdot \ell \cdot x$ (effective) steps is at most $e^{-x}$.
		\end{lem} 
		Using the upper bound of Theorem \ref{tau_upbound} as $\ell$, and by upper bounding $e^{-x}$ by $1/8N$ (or equivalently, lower bounding $x$ by $\ln (8N)$), the above lemma tells us that a single run of the process' simulation reaches absorption after $T = e \frac{r}{r-c} n^2 \ln (8N))$ steps with probability at most $1/8N$. By the union bound, the probability that any of the $N$ runs will reach absorption time after $T$ steps is upper bounded by $1/8$.	
%		Now, by using Theorem \ref{tau_upbound} and Markov's inequality, the process reaches absorption within $t$ steps with probability at least $1 - \epsilon$, for any $\epsilon \in (0,1)$ and any $t \geq \frac{r}{r-c} n^2 \frac{1}{\epsilon}$. Therefore,	the event that any individual simulation has not reached absorption within $T$ steps, happens with probability at most $1/(8N)$. By taking the union bound, the event of a simulation run not reaching absorption within $T$ steps happens with probability at most $1/8$. 

		Putting everything together, by the union bound of the event where the absolute difference of $g$ and $f$ is bounded away by $\epsilon f$ and the event of returning an error value (when absorption time exceeds $T$), we have that the probability of an output $g$ that is not as required is at most $1/4$ (by the latter bound and \eqref{eq:err_bound}, respectively). Thus, the probability of producing an output $g$ as required, is at least $3/4$.  
	\end{proof}

	%\section{Experimental Results}
	%
	%By running numerical experiments with graphs up to $n=100$ nodes, when the mutants pick the Clique graph, we could not find any directed (or undirected) graph for the residents that drags the fixation probability below $f_M$. It is also worth mentioning that suppressors of the standard 1-graph generalized Moran process, such as the directed line, the directed star or the undirected clique-wheel graph of \cite{NR13}, when picked by the residents, do not manage to yield a lower than $f_M$ fixation probability. The conjecture is that the residents cannot do any better than obtaining the extinction probability that the Moran process yields, when dealing with a Clique of the mutants.
	%
	%Also, experiments with graphs up to $n=100$ nodes indicate that directed (or undirected) regular graphs of arbitrary degree $d_1$ versus directed (or undirected) regular graphs of arbitrary degree $d_2$ yield fixation probability close to that of the Moran process, even for small graphs (e.g. for $n<15$). The conjecture is that when mutants pick an arbitrary regular graph, the fixation probability cannot be dragged significantly below $f_M$. In other words, every regular mutant graph provides the mutants with a fixation probability close to that of the Moran process which is relatively high enough.
	%

	\section{Conclusions and Open Problems}
	
	In this work we studied a generalization of the Moran process where the type (mutant or resident) of an agent determines, not only its fitness, but its neighbours as well. The latter seems to be a fundamental property of the evolutionary processes in cell structures, computer networks or social networks, but has never been considered before. 
	
	We gave some preliminary general results for the two-graphs Moran process and extended well known results from the typical (single-graph) Moran process \cite{L05}. Subsequently, we proposed a 2-player game theoretic view of the process by assuming that each type of individuals corresponds to a selfish player and then relating selfish costs with fixation and extinction probabilities. In this setting, we considered the extreme case of a mutant with maximum unpredictability, i.e. a mutant that can reproduce onto any other individual, and we examined the fixation probability that various resident structures yield against it. Our results seem to indicate that the clique is a dominant strategy for the mutant in the above game, however, a complete proof of that conjecture is needed and it is the main open problem stemming from this work. In fact, even the easier task of determining best responses for each player given specific strategies (graphs) from the other is a major open problem.
	
	Finally, we considered the problem of computing fixation probabilities, for arbitrary pairs of undirected graphs. We showed that there exists a pair of graph families that result in exponentially small fixation probability, a fact that excludes polynomial-time approximation via a method similar to \cite{D14} of the fixation probability in the general case. However, we showed explicitly how the method of \cite{D14} applies in the special case where the mutant graph is complete, by providing an FPRAS.
	
	For future work, an appealing problem to solve would be to find (if they exist) pure Nash equilibria in the aforementioned underlying game. This would imply that evolution through this process is drawn to specific pairs of graphs. Consequently, another interesting question would be: what is the time needed for the game to reach such a Nash equilibrium?

	\section*{Acknowledgements}
	Part of this work was done while the first author was affiliated with the University of Liverpool. The work of the first author was partially supported by the Alexander von Humboldt Foundation with funds from the German Federal Ministry of Education and Research (BMBF). The work of the third author was supported by the Greek State Scholarship Foundation. The work of the fourth author was partially supported by the ERC Project ALGAME. The authors would like to thank the anonymous referees for their detailed and constructive feedback.

	%%%%%%%%%%%%%%%%%%%%%%%%%%%%%%%%%%%%%%%%%%%%%%%%%%
	%%%%%%%%%%%%%%%%%%%%%%%%%%%%%%%%%%%%%%%%%%%%%%%%%%

%	\section*{References}
	
	\bibliography{references}
	
\end{document}